\documentclass[onecolumn,12pt]{IEEEtran} 
\usepackage{tikz}\usetikzlibrary{shapes,arrows.meta,matrix,calc}
\usetikzlibrary{patterns,snakes}
\usepackage{pgfplots}

\usepackage[pass, paperwidth=8.5in, paperheight=11in, margin=1in]{geometry}
\usepackage{amsmath,amssymb,bm,amsthm,mathtools,dsfont}
\usepackage{graphicx}
\usepackage{caption}
\usepackage[labelformat=simple]{subcaption}

\usepackage[hidelinks]{hyperref}  
\usepackage{cleveref}
\usepackage{cite}   
\usepackage{enumerate}
\usepackage{balance}   
\allowdisplaybreaks[2]  
\usepackage{float}

\begin{document}
\title{New Results on the Storage-Retrieval Tradeoff in Private Information Retrieval Systems}   
\author{
	\IEEEauthorblockN{Tao Guo, Ruida Zhou, and Chao Tian} 
	\thanks{Tao Guo is with the Department of Electrical and Computer Engineering, the University of California, Los Angeles, CA, USA. (e-mail: guotao@ucla.edu)
		
		Ruida Zhou and Chao Tian are with the Department of Electrical and Computer Engineering, Texas A\&M University, College Station, TX, USA. (e-mail: ruida@tamu.edu, chao.tian@tamu.edu)}
	}
	\maketitle
	\newcommand{\reffig}[1]{Figure \ref{#1}}
	\newcommand{\cA}{\mathcal{A}}
	\newcommand{\cB}{\mathcal{B}}
	\newcommand{\cC}{\mathcal{C}}
	\newcommand{\cK}{\mathcal{K}}
	\newcommand{\cL}{\mathcal{L}}
	\newcommand{\cN}{\mathcal{N}}
	\newcommand{\cP}{\mathcal{P}}
	\newcommand{\cQ}{\mathcal{Q}}
	\newcommand{\cR}{\mathcal{R}}
	\newcommand{\cS}{\mathcal{S}}
	\newcommand{\cT}{\mathcal{T}}
	\newcommand{\cU}{\mathcal{U}}
	\newcommand{\cX}{\mathcal{X}}
	\newcommand{\cY}{\mathcal{Y}}
	\newcommand{\bC}{\mathbb{C}}
	\newcommand{\bF}{\mathbf{F}}
	
	\newcommand{\dunderline}[1]{\underline{\underline{#1}}}
	
	\theoremstyle{plain}
	\newtheorem{theorem}{Theorem}
	\newtheorem{lemma}{Lemma}
	\newtheorem{corollary}{Corollary}
	\newtheorem{proposition}{Proposition}
	\newtheorem{conjecture}{Conjecture}
	\newtheorem{claim}{Claim}
	\newtheorem{example}{Example}
	\newtheorem{application}{Application}
	
	\theoremstyle{remark}
	\newtheorem{remark}{Remark}
	
	\theoremstyle{definition}
	\newtheorem{definition}{Definition}

\begin{abstract}
	In a private information retrieval (PIR) system, the user needs to retrieve one of the possible messages from a set of storage servers, but wishes to keep the identity of requested message private from any given server. 
	Existing efforts in this area have made it clear that the efficiency of the retrieval will be impacted significantly by the amount of the storage space allowed at the servers.  In this work, we consider the tradeoff between the storage cost and the retrieval cost. We first present three fundamental results: 1) a regime-wise 2-approximate characterization of the optimal tradeoff, 2) a cyclic permutation lemma that can produce more sophisticated codes from simpler ones, and 3) a relaxed entropic linear program (LP) lower bound that has a polynomial complexity. Equipped with the cyclic permutation lemma, we then propose two novel code constructions, and by applying the lemma, obtain new storage-retrieval points. Furthermore, we derive more explicit lower bounds by utilizing only a subset of the constraints in the relaxed entropic LP in a systematic manner. Though the new upper bound and lower bound do not lead to a more precise approximate characterization in general, they are significantly tighter than the existing art.	
\end{abstract}


\section{Introduction}
The analysis of private information retrieval (PIR) systems from the information-theoretic perspective has drawn significant attention recently \cite{chor-95,chor-98,shah2014one,Sun-Jafar-PIR-capacity-2017IT,Tian-Sun-Chen-PIR-IT19,Banawan-Ulukus-codedPIR-2018IT,uncoded1,uncoded2,Banawan-Ulukus-ITW19,sun2019breaking,tian2018shannon,zhou2019capacity,tian2019storage,lin2018mds,distributedPIR,coded-PIR--TChan-2015ISIT,Sun-Jafar-PIR-2017TIFS,PIR-dis-storage-18TIFS,PIR-MDS-dis-storage-2018IT,PIR-dis-storage-MDS-19IT,Sun-Jafar-SPIR-19IT,weakly-PIR-isit2019,Guo-Zhou-Tian-TIFS20,Wangqiwen-PIR-SPIR-Eva-2019IT,Wangqiwen-PIR-eva-2019IT,Wangqiwen-SPIR-MDS-2019IT,leaky-PIR-isit2019,WeakPIR-Zhou-Guo-Tian-isit20,Sun-Jafar-robustPIR-2017IT,Banawan-byzatine-colluding-IT19,Sun-Jafar-conjectured-2018IT,Banawan-Ulukus-PIR-MultiMessage-18IT,Sun-Jafar-PIR-multiround-2018IT,PIR-array-code-2019IT,PIR-cache-Allerton17}. 
The canonical model, where the messages are allowed to replicate over all the servers, was studied extensively and well-understood. Particularly, the capacity of the canonical PIR system was characterized recently by Sun and Jafar \cite{Sun-Jafar-PIR-capacity-2017IT}, and a more efficient code construction was presented in \cite{Tian-Sun-Chen-PIR-IT19}. 

Full replication of the messages at the storage servers can be costly, and the messages can be stored more efficiently by utilizing better storage codes. However, the amount of storage allowed at the servers will impact the efficiency of the retrieval. At one extreme, when the messages are replicated across all the servers, the retrieval can be made the most efficient; on the other hand, when no storage redundancy is allowed, the only possible strategy is to retrieve every message and thus highly inefficient. 

There has been increasing interest in understanding the storage-retrieval tradeoff in PIR systems. Banawan and Ulukus \cite{Banawan-Ulukus-codedPIR-2018IT} considered the case when each message is encoded by a maximum distance separable (MDS) code and stored across the servers, referred to as the MDS-PIR code, and characterized the capacity of this system. Sun and Tian presented two sets of codes where the messages are MDS-code that can beat the capacity of the separate MDS-PIR capacity by using joint storage coding for certain specific parameters \cite{sun2019breaking}.  Attia et al. considered the case when the storage servers can only store uncoded segments of the messages \cite{uncoded1,uncoded2},  and derived the full storage-retrieval tradeoff in such systems. A generalized code construction unifying the two codes was presented more recently in \cite{Banawan-Ulukus-ITW19}. Mathematically, we use $\alpha$ to denote the normalized average storage per server per message bit, and $\beta$ for the normalized average download cost per server by message bit (the precise definitions are given in \Cref{section-formulation}). In this context, the MDS-PIR code in \cite{Banawan-Ulukus-codedPIR-2018IT} achieve the following tradeoff points  
\begin{align}
\quad (\alpha, \beta) = \left( \frac{K}{T}, \frac{1}{N} \left(\sum_{i=0}^{K-1} \left(\frac{T}{N}\right)^i \right)\right), \quad T = 1, 2, \cdots, N,\label{eqn:mds}
\end{align}
the uncoded storage PIR code \cite{uncoded1,uncoded2} achieves the following tradeoff points
\begin{align}
\quad (\alpha, \beta) = \left( \frac{KT}{N}, \frac{1}{N} \left(\sum_{i=0}^{K-1} \frac{1}{T^i} \right)\right), \quad T = 1, 2, \cdots, N \label{eqn:uncoded}
\end{align}
and the unified code in \cite{Banawan-Ulukus-ITW19} achieves 
\begin{align}
(\alpha,\beta)=\left(\frac{KT_2}{NT_1},\frac{1}{N}\sum_{i=0}^{K-1}\left(\frac{T_1}{T_2}\right)^i\right), \quad T_1,T_2\in\{1, 2, \cdots, N\}, T_1\leq T_2. \label{eqn:gMDS}
\end{align}

Though significant progress has been made in these case where structural restrictions are placed on the storage codes, our understanding on the fundamental tradeoff between the storage cost and the retrieval cost is quite limited when these restrictions are removed. In fact, even for the smallest case with two servers and two messages, this tradeoff is not known. A Shannon-theoretic approach \cite{tian2018shannon} was used on this special case to improve the storage and download efficiency, and very specialized lower bounds were also given. Two general lower bounds were further given in \cite{tian2019storage} which focus on the two extreme points of the tradeoff curve. 

In this work, we studied the tradeoff between the storage cost and the retrieval cost in PIR systems without any structural storage restrictions. Firstly, three fundamental results are presented
\begin{enumerate}
\item A regime-wise 2-approximate characterization of the optimal tradeoff: The overall tradeoff can be partitioned into two regimes, where 2-approximation hods for either the storage cost or the retrieval cost.
\item A cyclic permutation lemma that can produce more sophisticated codes from simpler ones: This is a general technique, and it can be shown that uncoded storage PIR code \cite{uncoded1,uncoded2}  can be obtained directly from the code in \cite{Sun-Jafar-PIR-capacity-2017IT} with this lemma, and the generalized MDS-PIR code \cite{Banawan-Ulukus-ITW19} can be obtained from that in \cite{Banawan-Ulukus-codedPIR-2018IT}.
\item A relaxed entropic linear program (LP) lower bound that has a polynomial complexity: The generic entropic LP frame work \cite{yeung1997framework,tian2014characterizing,tian2018symmetry} may be used to compute lower bounds in this problem, which however has exponential numbers of variables and constraints. By utilizing the specific structure in the PIR problem, we select a subset of these inequalities and formulate a simpler LP that is more amicable for computation.
\end{enumerate}

With these results, we further seek to find improved upper bounds and lower bounds. We propose two novel code constructions, and by applying the cyclic permutation lemma, obtain a set of new storage-retrieval points. Then we derive a close-form lower bound by utilizing only a subset of the constraints in the relaxed entropic LP in a systematic manner. As a byproduct, we in fact obtain a set of lower bounds parametrized by a set of real values. Though the new upper bound and lower bound do not lead to a more precise approximate characterization in general, they are significantly tighter than the existing art.	

The rest of the paper is organized as follows. 
We formally define the problem in \Cref{section-formulation}. 
The three fundamental results on the optimal tradeoff are presented in \Cref{section-basic-result}. 
\Cref{section-UpperBounds} is mostly devoted to two new code constructions. A lower bound for the optimal tradeoff is then presented in \Cref{section-LowerBounds}, with some numerical results. We conclude the paper in \Cref{section-conclusion}. 
Some technical proofs are given in the appendices. 


\section{Problem Formulation}\label{section-formulation}
We adopt the notation $[i:j]\triangleq \{i,i+1,\ldots,j\}$ when $i\leq j$, and define it to be $\emptyset$ if $i>j$; the brackets will be omitted when appeared in subscripts. 
An $(N,K)$ {\it private information retrieval} (PIR) system can be described as follows. 
A total of $K$ mutually independent equal-length messages $W_{1:K}=(W_1, W_2, \cdots, W_K)$ are coded and stored in $N$ servers; the stored content at server $n$ is denoted as $S_n$. 
When retrieving message $W_k$, the user sends a query $Q^{[k]}_n$ to  server $n$, from which an answer $A_{n}^{[k]}$ was returned. 
After collecting the answers $A_{1:N}^{[k]}$ from all the servers, the user will recover the desired message $W_k$. 
The privacy requirement stipulates that any single server cannot derive any knowledge on the identity of the requested message based on the received query. 
In this work, we aim to study the tradeoff between the size of storage contents $S_{1:N}$ and that of the answers $A_{1:N}$.

Mathematically, a PIR system can almost be fully represented using information measures of involved random variables alone.
Each message $W_k$ $(k\in[1:K])$ is comprised of $L$ i.i.d. symbols uniformly distributed over a finite alphabet $\cX$. 
In $\log_{|\cX|}$-ary units, this is equivalent to 
\begin{align}
H(W_{1:K}) &= \sum_{k=1}^K H(W_k),\\ 
H(W_k) &= L,~k\in[1:K].
\end{align}
There are a total of $N$ servers, and each can store coded or uncoded contents of the messages, which is equivalent to the condition that the stored content $S_n\in\cS_n$ at server $n$  satisfies
\begin{align}
H(S_n | W_{1:K}) = 0, ~n\in[1:N].
\end{align}
A user aims to retrieve a message $W_k$, $k\in[1 : K]$ from the $N$ severs without revealing the identity $k$ to any individual server.
A random key $\bF$ is used to generate queries $Q_{1:N}^{[k]}=(Q_1^{[k]},Q_2^{[k]},\ldots,Q_N^{[k]})$, where $Q_n^{[k]} \in \mathcal{Q}_n$ for $n\in[1:N]$, which can be represented as
\begin{align}
H(Q_1^{[k]},Q_2^{[k]},\ldots,Q_N^{[k]}|\bF)=0,\quad k=1,2,\ldots,K. 
\end{align}
The random key is independent of messages, i.e., 
\begin{align}
I(\bF;W_{1:K})=0.
\end{align}
Server-$n$ uses the stored content $S_n$ and the query $Q_n^{[k]}$ to construct an answer $A_n^{[k]}$, and then sends the answer to the user, which is represented by the relation
\begin{align}
H(A_n^{[k]} | Q_n^{[k]}, S_n) = 0,  ~n\in[1:N], ~k\in[1:K].
\end{align}
The answer symbols are in a finite alphabet $\cY$, i.e., $A_n^{[k]}\in \cY^{\ell_n}$, where $\ell_n$ is the length of the answer.
With the answers from all servers $A_{1:N}^{[k]}$, together with queries $Q^{[k]}_{1:N}$ and the identity of the desired message $k$, the user can recover the desired message $W_k$, i.e., 
\begin{align}
H(W_k | A_{1:N}^{[k]}, Q^{[k]}_{1:N}) = 0.
\end{align}
The privacy requirement is more suitable to be represented using probability distribution relations, instead of information measures\footnote{Strictly speaking, it is possible to represent the privacy condition by introducing another random variable $\theta$ to represent the (random) index of the requested message, assuming the probability distribution of $\theta$ is known. The privacy requirement as represented by the probability distribution relations is more general, in the sense that there is no need to require the knowledge of the probability distribution of $\theta$.}, i.e., for any $q\in\cQ_n$  
\begin{align}
\Pr(Q_n^{[k]}=q)=\Pr(Q_n^{[k']}=q), ~\text{for any } k \not= k' \in [1:K]. 
\end{align}

The {\it operational} normalized average storage cost and the {\it operational} normalized average download cost are defined as
\begin{align}
\bar{\alpha} &\triangleq \frac{1}{NL} \sum_{n=1}^N \log_{|\cX|}|\cS_n|,   \\
\bar{\beta} &\triangleq \frac{\log_{|\cX|}|\cY|}{NL} \sum_{n=1}^N \mathbb{E}(\ell_n), 
\end{align}
which are the average amount of stored data per symbol of individual message and the expected amount of average downloaded data per symbol of desired message, respectively. 
In the sequel, we shall simply refer to them as the storage cost and download cost, respectively. 
Note that $\bar{\beta}$ does not depend on the value of $k$, since the random variable $\ell_n$ has an identical distribution for all $k\in[1:K]$ due to the privacy requirement. 

We say the storage-retrieval tradeoff point $(\alpha, \beta)$ is achievable, if there exists a PIR code whose operational storage cost $\bar{\alpha}$ and download cost $\bar{\beta}$ satisfy $\alpha\geq \bar{\alpha}$ and $\beta\geq \bar{\beta}$, respectively. 
The aim of this work is to characterize the set of all achievable pairs $(\alpha,\beta)$, or in other words, the optimal tradeoff between $\alpha$ and $\beta$. 
It is clear that 
\begin{align}
\bar{\alpha} & \geq   \frac{1}{NL} \sum_{n=1}^N H(S_n) \\
\bar{\beta} & \geq \frac{1}{NL} \sum_{n=1}^N H(A^{[k]}|Q^{[k]}_{1:N}), 
\end{align}
the right hand sides of which are referred to as the {\it informational} normalized storage cost and {\it informational} normalized download cost, respectively. We shall use the informational costs as surrogates for the operational costs in the rest of this work in order to derive meaningful lower bounds. A detailed discussion of these two definitions and their differences can be found in \cite{tian2019storage}. 

For some fixed download cost $\beta$, let $\alpha_{\min}(\beta)$ denote the minimum achievable storage cost for the download cost $\beta$, 
and $\beta_{\min}(\alpha)$ is defined similarly. 
It was established in \cite{Sun-Jafar-PIR-capacity-2017IT} that 
\begin{align}
\beta_{\min}(\infty)=\beta_0 \triangleq \frac{1}{N}+\frac{1}{N^2}+\ldots+\frac{1}{N^K},  \label{def-beta0}
\end{align}
and it is trivial to see
\begin{align}
\alpha_{\min}(\infty)=\alpha_0\triangleq \frac{K}{N},  \label{def-alpha0}
\end{align}
in order for the system to allow correct message retrieval. In fact the result in \cite{Sun-Jafar-PIR-capacity-2017IT} implies that $\beta_{\min}(K)=\beta_0$, and it is not difficult to verify $\alpha_{\min}(\alpha_0)=\alpha_0$. 

\begin{remark}
	The definitions of $\bar{\alpha}$ and $\bar{\beta}$ are consistent with the ``worst-case" definitions, which are
	\begin{align}
	\alpha_{\text{worst}}&\triangleq \max_{n \in [1:N]} \frac{\log_{|\cX|}|\cS_n|}{L}, \\  
	\beta_{\text{worst}}&\triangleq \max_{n \in [1:N]} \frac{\log_{|\cX|}|\cY|\mathbb{E}(\ell_n)}{L}.
	\end{align}
	This is because for any code that achieves the storage-retrieval tradeoff point $(\alpha, \beta)$, we can use space-sharing to construct a new code such that $(\alpha_{\text{worst}}, \beta_{\text{worst}}) = (\alpha,\beta)$.
\end{remark}

\section{Three Fundamental Results}\label{section-basic-result}
We first present three results that are not difficult from a technical point of view, but are of significant fundamental or instrumental importance. The first is a simple approximate characterization of the optimal $(\alpha,\beta)$ tradeoff, the second is a simple lemma which uses cyclic permutation to build more sophisticated codes from simpler ones, and the last is an extracted (low-complexity) linear programming lower bound that captures the most important constraints in the problem setting. 

\subsection{A Simple Approximate Characterization}

The following proposition provides a simple approximate characterization of the achievable storage-retrieval tradeoff. 
\begin{proposition}[Regime-wise 2-approximation]\label{thm-simple-approx}
For any $(N,K)$ PIR system where $N\geq 2$, 
	\begin{enumerate}[(i)]
	      \item The tradeoff point $(2\alpha_0,2\beta_0)$ is achievable; 
		\item Conversely, any achievable $(\alpha,\beta)$ much satisfy $\alpha\geq \alpha_0$ and $\beta\geq \beta_0$. 
	\end{enumerate}
\end{proposition}

\newcommand{\cross}{$\mathbin{\tikz [x=1.4ex,y=1.4ex,line width=.2ex] \draw (0,0) -- (0.5,0.5) (0,0.5) -- (0.5,0);}$}
\newcommand{\circles}{$\mathbin{\tikz [x=1.4ex,y=1.4ex,line width=.2ex] \filldraw (0,0) circle (1.0pt);}$}
\begin{figure}[!h]
	\centering
	
	\begin{subfigure}{0.45\textwidth}
		\centering
		\begin{tikzpicture}
		\draw [->] (0,0)--(4.5,0); \node [right] at (4.5,0) {$\alpha$}; 
		\draw [->] (0,0)--(0,3.5); \node [right] at (0,3.5) {$\beta$}; 
		
		\fill [gray,opacity=0.3] (0.5,0.5)--(0.5,3.0)--(1.0,3.0)--(1.0,1.0)--(4.0,1.0)--(4.0,0.5)--cycle;
		
		\draw [dashed] (0,0.5)--(4.0,0.5); \node [left] at (0,0.4) {$\beta_0$}; \draw [dashed] (0.5,0)--(0.5,3.0); 
		\draw [dashed] (1.0,1.0)--(4.0,1.0); \node [below] at (0.45,0) {$\alpha_0$}; \draw [dashed] (1.0,1.0)--(1.0,3.0); 
		\draw (1,0)--(1,0.1); \node[below] at (1.2,0) {$2\alpha_0$}; 
		\draw (0,1)--(0.1,1); \node[left] at (0,1.1) {$2\beta_0$}; 
		\node [left] at (0,3) {$\alpha_0$}; \draw [dashed] (0,3)--(0.1,3); 
		\node [below] at (4,0) {$K$}; \draw [dashed] (4,0)--(4,0.1); 
		\draw [red] plot [samples=100, domain=0.53 : 4] (\x,{1/(5*(\x-0.45)) + 0.48});
		\end{tikzpicture}
		\caption{Simple approximation}
		\label{fig_simple-tradeoff}
	\end{subfigure}
	~
	\begin{subfigure}{0.45\textwidth}
		\centering
		\begin{tikzpicture}
		\draw [->] (0,0)--(4.5,0); \node [right] at (4.5,0) {$\alpha$}; 
		\draw [->] (0,0)--(0,3.5); \node [right] at (0,3.5) {$\beta$}; 
		
		\draw [dashed] (0,0.5)--(4.0,0.5); \node [left] at (0,0.4) {$\beta_0$}; \draw [dashed] (0.5,0)--(0.5,3.0); 
		\draw [dashed] (1.0,1.0)--(4.0,1.0); \node [below] at (0.45,0) {$\alpha_0$}; \draw [dashed] (1.0,1.0)--(1.0,3.0); 
		\draw (1,0)--(1,0.1); \node[below] at (1.2,0) {$2\alpha_0$}; 
		\draw (0,1)--(0.1,1); \node[left] at (0,1.1) {$2\beta_0$}; 
		\node [left] at (0,3) {$\alpha_0$}; \draw [dashed] (0,3)--(0.1,3); 
		\node [below] at (4,0) {$K$}; \draw [dashed] (4,0)--(4,0.1); 
		
		\draw [blue] (0.5,3)--(0.58,2.2)--(0.66,1.55)--(0.75,1.05)--(0.85,0.93)--(1,0.84)--(1.25,0.75)--(1.5,0.7)--(2.5,0.55)--(4,0.5);
		\fill [gray,opacity=0.3] (0.5,0.5)--(0.5,3)--(0.58,2.2)--(0.66,1.55)--(0.75,1.05)--(0.85,0.93)--(1,0.84)--(1.25,0.75)--(1.5,0.7)--(2.5,0.55)--(4,0.5)--cycle;
		
		\foreach \m in {(0.5,3), (0.58,2.2), (0.66,1.55), (0.75,1.05), (0.85,0.93), (1,0.84), (1.25,0.75), (1.5,0.7), (2.5,0.55), (4,0.5)}
		\node [blue] at \m {\cross};
		
		\draw (1.3,2.5) rectangle (5.2,3.0); 
		\draw [blue] (1.4,2.75)--(1.9,2.75); \node [blue] at (1.65,2.75) {\cross}; \node at (3.55,2.75) {\small generalized MDS-PIR}; 
		\end{tikzpicture}
		\caption{Approximation vs. the generalized MDS-PIR bound}
		\label{fig_InnerBounds}
	\end{subfigure}
	
	\caption{Lower bound, upper bounds and simple approximation of the optimal tradeoff.}
\end{figure}
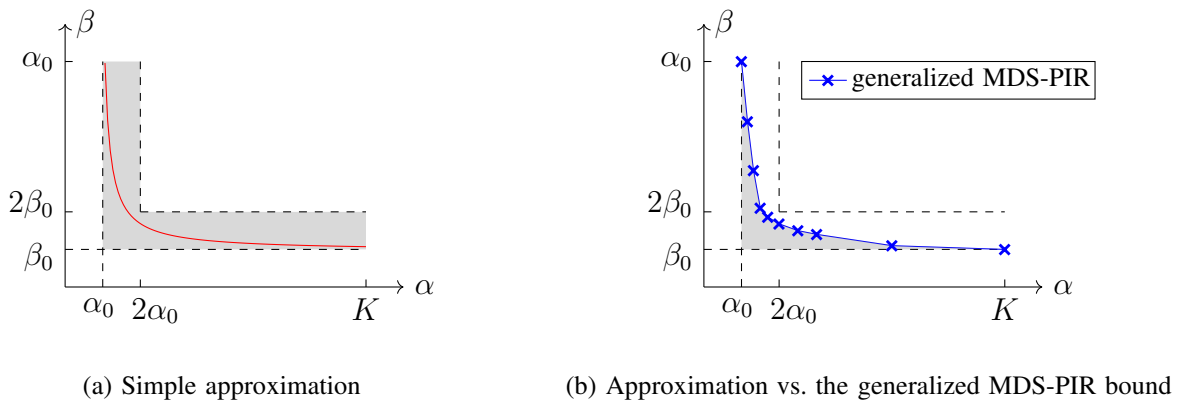

The combination of the upper bound given in (i) and the lower bound given in (ii) provide an approximate characterization as shown in Fig.~\ref{fig_simple-tradeoff}. In Fig.~\ref{fig_InnerBounds}, we further include the upper bounds induced by the generalize MDS-PIR code \cite{Banawan-Ulukus-ITW19} to illustrate this approximate characterization. 

\begin{proof}[Proof of \Cref{thm-simple-approx}]
	The lower bounds in (ii) follow simply from the definition of $\alpha_0$ and $\beta_0$, and thus we only need to prove the upper bounds in (i). 
	This can be done by showing that the point $\left(2\alpha_0,2\beta_0\right)$ is above the tradeoff curve achieved by the 
	uncoded storage PIR code given in \cite{uncoded1,uncoded2}, for which the storage cost and download cost are given by the lower convex envelop of the following points
	\begin{align}
	\quad (\bar{\alpha}, \bar{\beta}) = \left( \frac{KT}{N}, \frac{1}{N} \left(\sum_{i=0}^{K-1} \frac{1}{T^i} \right)\right), \quad T = 1, 2, \cdots, N. \label{uncoded-tradeoff}
	\end{align}
	Taking $T=2$, we obtain 
	\begin{align}
	\quad (\bar{\alpha}, \bar{\beta}) = \left(2\alpha_0, \frac{1}{N} \cdot \frac{1-\left(\frac{1}{2}\right)^K}{1-\frac{1}{2}}\right). 
	\end{align}
	For $N\geq2$, it is seen that the 
	\begin{equation}
	\bar{\beta}=\frac{1}{N} \frac{1-\left(\frac{1}{2}\right)^K}{1-\frac{1}{2}} < \frac{2}{N} \cdot \frac{1-\left(\frac{1}{N}\right)^K}{1-\frac{1}{N}}=2\beta_0, 
	\end{equation}
	which completes the proof.
\end{proof}

\begin{remark}
It is also possible to utilize the upper bound induced by the MDS-PIR code \cite{Banawan-Ulukus-codedPIR-2018IT} to prove this proposition, which we omit for brevity. 
\end{remark}

This approximate characterization shows that the storage-retrieval tradeoff can be divided into three regimes: a storage-bound regime, a retrieval-bound regime, and an intermediate regime. 
In the storage-bound regime $\beta\geq 2\beta_0$, the optimal storage cost is sandwiched between $\alpha_0$ and $2\alpha_0$ for any fixed  $\beta$;
in the retrieval-bound regime $\alpha\geq 2\alpha_0$, the optimal retrieval cost is sandwiched between $\beta_0$ and $2\beta_0$ for any fixed $\alpha$; in the intermediate regime, where $\alpha\leq 2\alpha_0$ and $\beta\leq 2\beta_0$, the optimal sum rate $\alpha+\beta$ is sandwiched between $\alpha_0+\beta_0$ and $2\alpha_0+2\beta_0$. Thus in the first regime, the potential loss of using the uncoded PIR code (or the MDS-PIR code) in terms of the storage cost is less than a multiplicative factor of 2, while in the second, the potential loss of using either of these two codes in terms of the retrieval cost is less than a factor of 2. In the intermediate regime, the sum-rate loss is less than a factor of 2 using these codes. This result makes clear what questions remain difficult: to find good approximate (or exact) characterization of the retrieval cost in the storage-bound regime, that on the storage cost in the retrieval-bound regime, and either direction in the intermediate regime. In \cite{tian2019storage}, these questions were considered for the extreme cases when $\alpha=\alpha_0$ and $\beta=\beta_0$, and a precise characterization was given for the former and an approximate one for the latter. However, beyond those two extreme cases, the answers to these questions remains elusive. In the sections to follow, we shall provide further results toward answering these questions. 

\subsection{A Cyclic Permutation Lemma}
We next introduce a general technique to produce more sophisticated codes from simpler codes, and present several immediate applications of this lemma. In Section \ref{section-UpperBounds} we shall further utilize this technique to produce other achievable $(\alpha,\beta)$ tradeoff points based on several new code constructions. 

\begin{lemma}[Cyclic permutation lemma]\label{lemma-cyclic-permutation} 
If an $(N,K)$ PIR code can achieve the tradeoff point $(\alpha, \beta)$, then there exists an $(M,K)$ PIR code, $M\geq N$, that achieves the tradeoff point $\left(\frac{N}{M}\alpha, \frac{N}{M}\beta \right)$.
\end{lemma}
\begin{proof}
We prove the lemma by generating an $(M,K)$ PIR code from an $(N,K)$ PIR code using round-robin, which is illustrated in Fig.~\ref{fig_round-robin}. 
Let the message length in the $(N,K)$ PIR (base) code be $L$, and in the $(M,K)$ code to be constructed, the message length will be $ML$. 	Therefore, for the $(M,K)$ PIR problem, we can partition each message $W_k$ into $M$ sub-messages $W_k^1, W_k^2, \cdots, W_k^M$, each has a message length $L$. For each $m\in[1:M]$, the sub-messages $W_1^m, W_2^m, \cdots, W_K^m$ can be encoded by the original $(N,K)$ code, and placed on a set of $N$ consecutive servers, i.e., the base $(N,K)$ PIR code is utilized on the servers in a round-robin manner. More precisely, for $m=1,2,\ldots,M$, the $K$ sub-messages $W_1^m, W_2^m, \cdots, W_K^m$  are encoded using the $(N,K)$ PIR storage code as $S_1^m, S_2^m, \cdots, S_N^m$; for notation simplicity, further define $S_n^m=\emptyset$ for $n\in[N+1,N+2,\cdots,M]$. Then server $m$ stores the encoded messages $S_{(m+1)_N}^1, S_{(m+2)_N}^2, \cdots, S_{(m+M)_N}^M$ for $m=1,2,\ldots,M$, where $(x)_N$ is defined for any integer $x$ as 
	\begin{equation}
	(x)_N = 
	\begin{cases}
	x \mod N, &\text{if }x\mod N\neq 0  \\
	N, &\text{if }x \mod N =0 
	\end{cases}.
	\end{equation}

	The retrieval is done on each group of sub-messages, $(W_1^m,W_2^m,\ldots,W_K^m)$, which are stored on the corresponding server set that they are stored, for $m=1,2,\ldots,M$. Since for each group of sub-messages, the retrieval is private by the property of the base $(N,K)$ PIR code, the overall retrieval is also private. 
	Since the message lengths of the base code and the new code are $L$ and $ML$, respectively, 
	the resulting storage and download cost of the new code are $\alpha' = \frac{MN\alpha L}{MML}=\frac{N}{M} \alpha$ and $\beta' = \frac{MN\beta L}{MML} = \frac{N}{M} \beta$. The proof is complete.
\end{proof}

\newcommand{\squares}{$\mathbin{\tikz [x=2.0ex,y=2.0ex,line width=.2ex] \filldraw (0,0) rectangle (1,1);}$}
\begin{figure}[!t]
	\centering
	\begin{tikzpicture}[scale=1.45]
	\draw (-3.2,0.2) rectangle (-3.7,-3.7); \node at (-3.5,-1.9) {\rotatebox[origin=c]{270}{Messages $W_{1:K}$}}; 
	\draw[->,>=stealth] (-3.2,0)--(-2.5,0); \draw[->,>=stealth] (-3.2,-0.5)--(-2.5,-0.5); \draw[->,>=stealth] (-3.2,-1.0)--(-2.5,-1.0); 
	\draw[->,>=stealth] (-3.2,-2.0)--(-2.7,-2.0); \draw[->,>=stealth] (-3.2,-2.5)--(-2.8,-2.5); \draw[->,>=stealth] (-3.2,-3.5)--(-2.5,-3.5); 
	
	\node [red, left] at (-0.2,0) {\small Sub-message $W_{1:K}^1$}; 
	\node[red] at (0,0) {\squares}; \node[red] at (0.5,0) {\squares}; \node[red] at (1.0,0) {$\bm{\cdots}$}; \node[red] at (1.5,0) {\squares};
	
	\node [cyan, left] at (-0.2,-0.5) {\small Sub-message $W_{1:K}^2$}; 
	\node[cyan] at (0.5,-0.5) {\squares}; \node[cyan] at (1.0,-0.5) {\squares}; \node[cyan] at (1.5,-0.5) {$\bm{\cdots}$}; \node[cyan] at (2.0,-0.5) {\squares};
	
	\node [blue, left] at (-0.2,-1.0) {\small Sub-message $W_{1:K}^3$}; 
	\node[blue] at (1.0,-1.0) {\squares}; \node[blue] at (1.5,-1.0) {\squares}; \node[blue] at (2.0,-1.0) {$\bm{\cdots}$}; \node[blue] at (2.5,-1.0) {\squares};
	
	\node [black, left] at (-0.5,-1.5) {$\vdots$}; \node[black] at (2.0,-1.5) {$\vdots$};
	
	\node [orange, left] at (-0.2,-2.0) {\small Sub-message $W_{1:K}^{M-N}$}; 
	\node[orange] at (2.0,-2.0) {\squares}; \node[orange] at (2.5,-2.0) {\squares}; \node[orange] at (3.0,-2.0) {$\bm{\cdots}$}; \node[orange] at (3.5,-2.0) {\squares};
	
	\node [green, left] at (-0.2,-2.5) {\small Sun-message $W_{1:K}^{M-N+1}$}; 
	\node[green] at (0.0,-2.5) {\squares}; \node[green] at (2.5,-2.5) {\squares}; \node[green] at (3.0,-2.5) {$\bm{\cdots}$}; \node[green] at (3.5,-2.5) {\squares};
	
	\node [black, left] at (-0.5,-3.0) {$\vdots$}; \node[black] at (2.5,-3.0) {$\vdots$};
	
	\node [blue!40, left] at (-0.2,-3.5) {\small Sub-message $W_{1:K}^{M}$}; 
	\node[blue!40] at (0.0,-3.5) {\squares}; \node[blue!40] at (0.5,-3.5) {$\bm{\cdots}$}; \node[blue!40] at (1.0,-3.5) {\squares};  \node[blue!40] at (3.5,-3.5) {\squares};
	
	
	\draw[dashed, black!30] (-0.17,0.25) rectangle (0.17,-3.8); \node[above] at (0,0.2) {\rotatebox[origin=c]{270}{\small Server 1}}; 
	\draw[dashed, black!30] (0.33,0.25) rectangle (0.67,-3.8); \node[above] at (0.5,0.2) {\rotatebox[origin=c]{270}{\small Server 2}}; 
	\draw[dashed, black!30] (0.83,0.25) rectangle (1.17,-3.8); \node[above] at (1.0,0.7) {$\cdots$}; 
	\draw[dashed, black!30] (1.33,0.25) rectangle (1.67,-3.8); \node[above] at (1.5,0.2) {\rotatebox[origin=c]{270}{\small Server $N$}}; 
	\draw[dashed, black!30] (1.83,0.25) rectangle (2.17,-3.8); \node[above] at (2.0,0.2) {\rotatebox[origin=c]{270}{\small Server $N+1$}}; 
	\draw[dashed, black!30] (2.33,0.25) rectangle (2.67,-3.8); \node[above] at (2.5,0.2) {\rotatebox[origin=c]{270}{\small Server $N+2$}}; 
	\draw[dashed, black!30] (2.83,0.25) rectangle (3.17,-3.8); \node[above] at (3.0,0.7) {$\cdots$}; 
	\draw[dashed, black!30] (3.33,0.25) rectangle (3.67,-3.8); \node[above] at (3.5,0.2) {\rotatebox[origin=c]{270}{\small Server $M$}}; 
	
	\node[above] at (4.25,0.2) {\rotatebox[origin=c]{270}{\small $(N,K)$ PIR code}}; 
	
	\draw[dashed, red!80] (-0.22,0.17) rectangle (3.75,-0.17); \draw[->,>=stealth] (3.75,0)--(4.0,0); \draw (4.0,0.2) rectangle (4.5,-0.2); 
	\draw[->,>=stealth] (4.5,0.0)--(4.75,0.0); \node[right, red!80] at (4.75,0.0) {\small Sub-code $1$};
	\draw[dashed, cyan!80] (-0.22,-0.33) rectangle (3.75,-0.67); \draw[->,>=stealth] (3.75,-0.5)--(4.0,-0.5); \draw (4.0,-0.3) rectangle (4.5,-0.7); 
	\draw[->,>=stealth] (4.5,-0.5)--(4.75,-0.5); \node[right, cyan!80] at (4.75,-0.5) {\small Sub-code $2$};
	\draw[dashed, blue!80] (-0.22,-0.83) rectangle (3.75,-1.17); \draw[->,>=stealth] (3.75,-1.0)--(4.0,-1.0); \draw (4.0,-0.8) rectangle (4.5,-1.2); 
	\draw[->,>=stealth] (4.5,-1.0)--(4.75,-1.0); \node[right, blue!80] at (4.75,-1.0) {\small Sub-code $3$};
	\node at (4.25,-1.5) {$\vdots$};
	\draw[dashed, orange!80] (-0.22,-1.83) rectangle (3.75,-2.17); \draw[->,>=stealth] (3.75,-2.0)--(4.0,-2.0); \draw (4.0,-1.8) rectangle (4.5,-2.2); 
	\draw[->,>=stealth] (4.5,-2.0)--(4.75,-2.0); \node[right, orange!80] at (4.75,-2.0) {\small Sub-code $M \! - \! N$};
	\draw[dashed, green!80] (-0.22,-2.33) rectangle (3.75,-2.67); \draw[->,>=stealth] (3.75,-2.5)--(4.0,-2.5); \draw (4.0,-2.3) rectangle (4.5,-2.7); 
	\draw[->,>=stealth] (4.5,-2.5)--(4.75,-2.5); \node[right, green!80] at (4.75,-2.5) {\small Sub-code $M \! - \! N+1$};
	\node at (4.25,-3.0) {$\vdots$};
	\draw[dashed, blue!30] (-0.22,-3.33) rectangle (3.75,-3.67); \draw[->,>=stealth] (3.75,-3.5)--(4.0,-3.5); \draw (4.0,-3.3) rectangle (4.5,-3.7); 
	\draw[->,>=stealth] (4.5,-3.5)--(4.75,-3.5); \node[right, blue!30] at (4.75,-3.5) {\small Sub-code $M$};
	
	
	\draw (7.5,0.2) rectangle (8.0,-3.7); \node at (7.8,-1.9) {\rotatebox[origin=c]{270}{\small $(M,K)$ PIR code}}; 
	\draw[->,>=stealth] (6.3,0)--(7.5,0); \draw[->,>=stealth] (6.3,-0.5)--(7.5,-0.5); \draw[->,>=stealth] (6.3,-1.0)--(7.5,-1.0); 
	\draw[->,>=stealth] (6.8,-2.0)--(7.5,-2.0); \draw[->,>=stealth] (7.0,-2.5)--(7.5,-2.5); \draw[->,>=stealth] (6.3,-3.5)--(7.5,-3.5); 
	\end{tikzpicture}
	\caption{Generation of $(M,K)$ PIR code from $(N,K)$ PIR code using round-robin.}
	\label{fig_round-robin}
\end{figure}
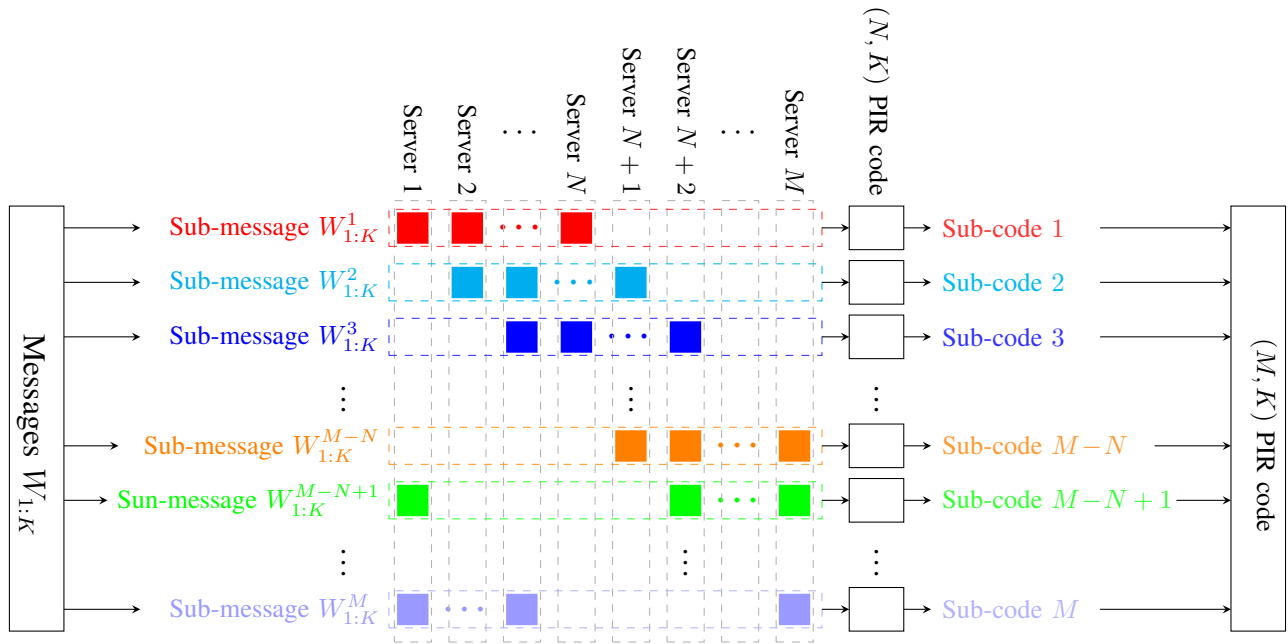

\begin{remark}
The lemma in fact also holds for the $T$-colluding PIR problem \cite{Sun-Jafar-robustPIR-2017IT} and the symmetric PIR problem \cite{Sun-Jafar-SPIR-19IT}, when the storage cost is taken into consideration.  The proof is identical and thus omitted here.
\end{remark}

We can apply the cyclic permutation lemma on any existing codes, e.g., the Sun-Jafar code \cite{Sun-Jafar-PIR-capacity-2017IT}, the TSC code~\cite{Tian-Sun-Chen-PIR-IT19}, the MDS-PIR code \cite{Banawan-Ulukus-codedPIR-2018IT}, and the uncoded storage PIR code \cite{uncoded1,uncoded2}. In fact, the performance of the uncoded storage PIR code in \cite{uncoded1,uncoded2} and the generalized MDS-PIR coded in \cite{Banawan-Ulukus-ITW19} can be obtained this way from the code \cite{Sun-Jafar-PIR-capacity-2017IT} and \cite{Banawan-Ulukus-codedPIR-2018IT}, respectively, as we shall show next.  

\begin{application}
	The $(M,K)$ uncoded storage PIR code in \cite{uncoded1,uncoded2} can be produced from the Sun-Jafar code \cite{Sun-Jafar-PIR-capacity-2017IT} using the cyclic permutation lemma. The  storage cost and download cost of $(N,K)$ Sun-Jafar code is 
	\begin{align}
	(\alpha,\beta)=\left(K,\frac{1}{N}+\frac{1}{N^2}+\cdots+\frac{1}{N^K}\right). 
	\end{align}
	By applying \Cref{lemma-cyclic-permutation} to an $(N,K)$ Sun-Jafar code, the corresponding storage download cost can be obtained as 
	\begin{align}
	(\alpha',\beta')=\frac{N}{M}(\alpha,\beta)=\left(\frac{NK}{M},\frac{1}{M}\sum_{i=0}^{K-1}\frac{1}{N^i}\right). 
	\end{align}
	For different storage requirement, we can choose different Sun-Jafar base code by varying $N$. By taking $N=1,2,\cdots,M$, the storage-retrieval tradeoff of the uncoded storage PIR code in \eqref{uncoded-tradeoff} is obtained. We remark that the base code can be any other PIR capacity-achieving code, e.g., the TSC code \cite{Tian-Sun-Chen-PIR-IT19}, which can yield the same performance. 
\end{application}

\begin{application}
	The $(M,K)$ generalized MDS-PIR code in \cite{Banawan-Ulukus-ITW19} can be produced using the cyclic permutation lemma from the MDS-PIR codes in~\cite{Banawan-Ulukus-codedPIR-2018IT}. 
	The average storage and download cost of MDS-PIR code with parameters $(N,K)$ is given in \eqref{eqn:mds} as 
	\begin{align}
	\quad (\alpha, \beta)=\left( \frac{K}{T}, \frac{1}{N} \left(\sum_{i=0}^{K-1} \frac{T^i}{N^i} \right)\right), \quad T = 1, 2, \cdots, N. 
	\end{align}
	By applying \Cref{lemma-cyclic-permutation} to an $(N,K)$ MDS-PIR code, the corresponding storage download cost can be obtained as 
	\begin{align}
	(\alpha,\beta)=\frac{N}{M}(\alpha,\beta)=\left(\frac{NK}{MT},\frac{1}{M}\sum_{i=0}^{K-1}\frac{T^i}{N^i}\right), \quad T = 1, 2, \cdots, N. 
	\end{align}
	By letting $N=1,2,\cdots,M$, we obtain the storage-retrieval tradeoff of the generalized MDS-PIR code \cite{Banawan-Ulukus-ITW19} given in (\ref{eqn:gMDS}). We illustrate this storage-retrieval points for the codes obtained by applying the technique when $N=K=3$ and $M=4$ in Fig.~\ref{fig_example-cyclic}, where it is seen that indeed new tradeoff points are obtained beyond those achieved by the MDS-PIR base code. 
	
		\begin{figure}[!t]
		\centering
		\includegraphics[scale=0.6]{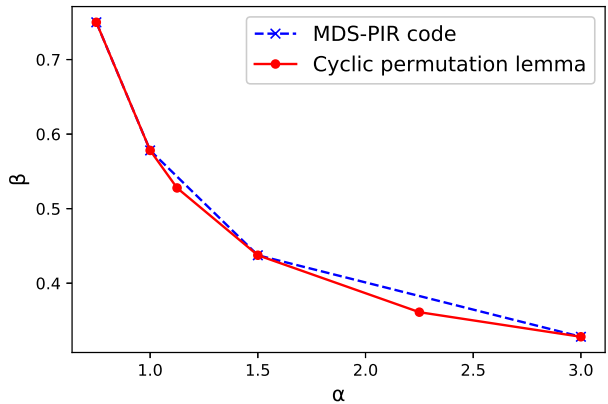}
		\caption{MDS-PIR code upper bound vs. that obtained by cyclic permutation lemma.}
		\label{fig_example-cyclic}
	\end{figure}
\end{application}

\subsection{An Linear Programming Lower Bound}\label{section-LP-LowerBound}
Characterizing fully the strorage-retrieval tradeoff appears difficult, partly due to the lack of strong lower bounds, though some initial effort was reported in \cite{Sun-Jafar-PIR-multiround-2018IT,tian2018shannon, tian2019storage}. This problem can potentially be solved computationally in the generic entropic linear programming (LP)  framework \cite{yeung1997framework}, similar to the approach discussed in \cite{tian2014characterizing,tian2018symmetry}. This generic approach however suffers from high complexity that is  exponential in the number of random variables in the problem, since the variables in this generic entropic LP are the joint entropy values of all the possible subsets of these random variables. On the other hand, the PIR problem in fact has a very special structure, which can be well captured by a small class of inequalities. In the following, we use one special class of inequalities to formulate a relaxed linear program. 

As a first step, we shall utilize the symmetry structure in this problem. As shown in \cite{Tian-Sun-Chen-PIR-IT19}, any PIR code can be symmetrized without sacrificing the storage and download cost to satisfy two symmetry relations: message symmetry and server symmetry. 
Let $\cP_m$ be the set of all permutations of $[1:m]$. A symmetrized PIR code satisfies the following condition for any $\cA, \cB \subset [1:N]$ and $\cC \subset [1:K]$, $k \in [1:K]$, and any $\pi \in \cP_N$ and $\pi' \in \cP_K$, 
\begin{align}
H(S_{\cA}, A^{[k]}_{\cB} |\bF,  W_{\cC}) = H(S_{\pi(\cA)}, A^{[\pi'(k)]}_{\pi(\cB)} |\bF, W_{\pi'(\cC)}).
\end{align}
In the sequel, we consider only such symmetrized codes without loss of optimality. 

For any nonnegative integers $a$ and $b$ such that $a + b \leq N$, let $\cA$ and $\cB$ be two disjoint subsets of $[1:N]$ with $|\cA|=a$ and $|\cB| = b$, we define (for any symmetrized code)
\begin{align*}
&x_k(a, b) \triangleq H(S_{\mathcal{A}}, A^{[k]}_{\mathcal{B}} |\bF,  W_{1:k})/L  \\
&y_k(a, b) \triangleq H(S_{\mathcal{A}}, A^{[k]}_{\mathcal{B}} |\bF,  W_{1:k-1})/L.
\end{align*}

It is straightforward to see that by definition
\begin{align*}
y_1(1, 0) &\geq   \alpha\\
y_1(0,1) &\geq  \beta.
\end{align*}
Thus in order to lower bound a linear combination of $a_0 \alpha+b_0 \beta$ where $a_0,b_0\geq 0$, we can consider the following linear program, which we summarize as a proposition.
\begin{proposition}[Relaxed entropic LP]
\label{prop:LP}
For any achievable $(\alpha,\beta)$, the linear combination $a_0 \alpha+b_0 \beta$ where $a_0,b_0\geq 0$ is lower-bounded by the solution of the following linear program.
\begin{eqnarray} \label{eqn:lp}
\text{minimize: } & & a_0 y_1(1, 0) + b_0 y_1(0,1) \\
\text{subject to: }
&\text{(Submodular)} & x_k( |\mathcal{A}_1|, |\mathcal{B}_1| ) + x_k(|\mathcal{A}_2|, |\mathcal{B}_2| ) \geq  x_k( |\mathcal{A}_1 \cup \mathcal{A}_2|, |\mathcal{B}_1 \cup \mathcal{B}_2 / (\mathcal{A}_1 \cup \mathcal{A}_2)| )  \notag \\
&&\qquad\qquad\qquad +x_k( |\mathcal{A}_1 \cap \mathcal{A}_2|, |\mathcal{B}_1 \cap \mathcal{B}_2|+ |\mathcal{A}_1 \cap \mathcal{B}_2| + |\mathcal{A}_2 \cap \mathcal{B}_1|  ), \notag\\
&& \qquad\quad  \forall k \in [1:K-1], \forall \cA_i, \cB_i \in [1:N], \cA_i\cap \cB_i=\emptyset, i=1,2\\ 
&& y_k( |\mathcal{A}_1|, |\mathcal{B}_1| ) + y_k(|\mathcal{A}_2|, |\mathcal{B}_2| ) \geq  y_k( |\mathcal{A}_1 \cup \mathcal{A}_2|, |\mathcal{B}_1 \cup \mathcal{B}_2 / (\mathcal{A}_1 \cup \mathcal{A}_2)| )  \notag \\
&&\qquad\qquad\qquad+y_k( |\mathcal{A}_1 \cap \mathcal{A}_2|, |\mathcal{B}_1 \cap \mathcal{B}_2|+ |\mathcal{A}_1 \cap \mathcal{B}_2| + |\mathcal{A}_2 \cap \mathcal{B}_1|  ), \notag\\
&& \qquad\qquad  \forall k \in [1:K], \forall \cA_i, \cB_i \in [1:N], \cA_i\cap \cB_i=\emptyset, i=1,2\\ 
& \text{(Monotone) }& x_k(a,b) \geq x_{k}(a,b-1), \notag\\
&& \qquad\qquad  \forall k \in [1:K-1], \forall a \in [0:N-1], \forall b \in [1:N - a]  \\ 
&&  y_{k}(a, b) \geq y_k(a,b-1),\notag\\
&& \qquad\qquad  \forall k \in [1:K], \forall a \in [0:N-1], \forall b \in [1:N - a]  \\ 
&\text{(Decodable)}& y_k(a, N-a) \geq 1 + x_k(a, N-a), ~ \forall k \in [1:K], \forall a \in [0:N-1]\\
&\text{(Han's inequality) }& y_{k}(0,b) \geq \frac{b}{N} y_{k}(0,N), \quad  \forall k \in [1:K], \forall b \in [1:N-1] \\
&\text{(Privacy)}&  x_k(a, 1) = y_{k+1}(a, 1), \quad  \forall k \in [1:K-1], \forall a \in [0:N-1]\\
&\text{(Invariance)}&  x_k(a,0) = y_{k+1}(a,0), \quad  \forall k \in [1:K-1], \forall a \in [1:N]  \\
&\text{(Boundary)}&x_K(a,b) = 0, \quad \forall a \in [0:N], \forall b \in [0: N - a].
\end{eqnarray}
\end{proposition}

In Fig. \ref{fig_compare-LP-Tian19}, we illustrate a set of bounds obtained by solving this linear program for the case $(N,K)=(5,3)$, which are considerably tighter than known bounds in the literature. 
The capacity bound $\beta\geq \beta_0=0.248$ is shown as a horizontal bound, which is indeed obtained through solving the relaxed entropic LP. The two constraints given in~\cite{tian2019storage} are also obtained through the relaxed entropic LP. This is not surprising, since the insights used to formulate the relaxed entropic LP are partly motivated by the proof steps used~there. 

\begin{figure}[tb]
	\centering
	\includegraphics[scale=0.6]{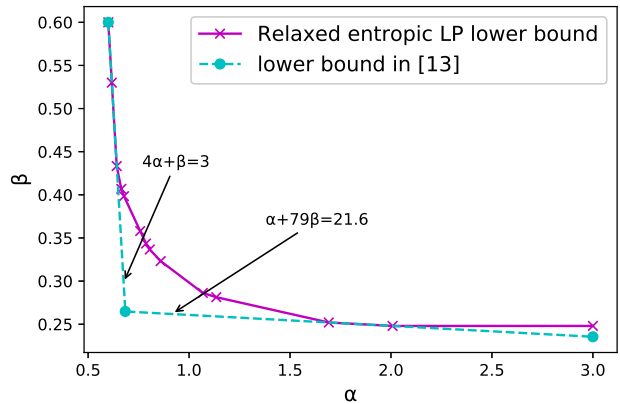}
	\caption{Comparison of the lower bounds in \Cref{prop:LP} and \cite{tian2019storage} for $(N,K)=(5,3)$.}
	\label{fig_compare-LP-Tian19}
\end{figure}

\begin{proof}
The constraints in this LP need to hold for any valid (symmetrized) PIR code for the reasons given below.
\begin{itemize}
\item Submodular: for any disjoint $\cA_1, \cB_1 \subset [1:N]$, disjoint $\cA_2, \cB_2 \subset [1:N]$ and any $k \in [0:K-1]$, $k' \in [1:K]$, by the submodular property of the entropy function, we have
\begin{align}
&H(S_{\cA_1}, A^{[k']}_{\cB_1} |\bF,  W_{1:k}) + H(S_{\cA_2}, A^{[k']}_{\cB_2} |\bF,  W_{1:k}) \notag\\
&\geq H( (S_{\cA_1}, A^{[k']}_{\cB_1}) \cup  (S_{\cA_2}, A^{[k']}_{\cB_2}) |\bF,  W_{1:k}) + H( (S_{\cA_1}, A^{[k']}_{\cB_1}) \cap  (S_{\cA_2}, A^{[k']}_{\cB_2}) |\bF,  W_{1:k})\notag \\
&=  H(S_{\cA_1 \cup \cA_2}, A^{[k']}_{\mathcal{B}_1 \cup \mathcal{B}_2 / (\mathcal{A}_1 \cup \mathcal{A}_2)} |\bF,  W_{1:k}) + H(S_{\cA_1 \cap \cA_2}, A^{[k']}_{(\mathcal{B}_1 \cap \mathcal{B}_2) \cup (\mathcal{A}_1 \cap \mathcal{B}_2) \cup (\mathcal{A}_2 \cap \mathcal{B}_1)} |\bF,  W_{1:k}).
\end{align}
The constraints on $y_k(a,b)$ hold for the similar reason. 
\item Monotone: for any disjoint $\cA, \cB \subset [1:N]$ with $\cB$ non-empty, let $\cB' \subset \cB$ with $|\cB'| = |\cB| - 1$, and for any $k \in [0:K-1]$ we have 
\begin{align}
H(S_{\mathcal{A}}, A^{[k]}_{\mathcal{B}} |\bF,  W_{1:k}) \geq H(S_{\mathcal{A}}, A^{[k]}_{\mathcal{B}'} |\bF,  W_{1:k}),   \label{entropicLP-monotone}
\end{align}
then the given linear constraints on $x_{k}(a,b)$ hold due to the symmetry relation mentioned earlier. The constraints on $y_k(a,b)$ hold for the same reason. 
\item Decodable: for any $\cA \subset [1:N]$, for any $k \in [1: K]$, we have
\begin{align}
H(S_{\mathcal{A}}, A^{[k]}_{\mathcal{A}^c} |\bF,  W_{1:k - 1}) = H(W_k, S_{\mathcal{A}}, A^{[k]}_{\mathcal{A}^c} |\bF,  W_{1:k - 1}) = H(W_k) + H(S_{\mathcal{A}}, A^{[k]}_{\mathcal{A}^c} |\bF,  W_{1:k}),
\end{align}
where $\mathcal{A}^c=[1:N]\setminus \mathcal{A}$.
\item Han's inequality: for any $b \in [1:N-1]$ and $k \in [1:K]$, by the conditional version of the Han's inequality, we have
\begin{align}
\frac{1}{b \binom{N}{b}} \sum_{\cB \subset [1:N]: |\cB| = b} H(A^{[k]}_{\mathcal{B}} |\bF,  W_{1:k-1}) \geq \frac{1}{N} H(A^{[k]}_{[1:N]} |\bF,  W_{1:k-1}).  \label{entropicLP-Han}
\end{align}
\item Privacy: due to the Markov string $\bf\leftrightarrow Q^{[k]}_n\leftrightarrow (A^{[k]}_n, W_{1:K}, S_{1:N})$, and the fact that $W_{1:K}$ is independent of $\bF$, we have that for any $\cA \subset [1:N]$ with $|\cA| < N$ and any $n \in \cA^c$, 
\begin{eqnarray}
H(S_{\mathcal{A}}, A^{[k]}_{n} |\bF,  W_{1:k}) &=& H(S_{\mathcal{A}}, A^{[k]}_{n} |Q_n^{[k]},  W_{1:k}) \notag\\
H(S_{\mathcal{A}}, A^{[k+1]}_{n} |\bF,  W_{1:k})&=&H(S_{\mathcal{A}}, A^{[k+1]}_{n} |Q_n^{[k+1]},  W_{1:k}).
\end{eqnarray}
By the privacy requirement, the distributions of $(Q^{[k]}_n,A^{[k]}_n,S_{\mathcal{A}},W_{1:k})$ and  $(Q^{[k+1]}_n,A^{[k+1]}_n,S_{\mathcal{A}},W_{1:k})$ are identical, and it follows that 
\begin{align}
H(S_{\mathcal{A}}, A^{[k]}_{n} |\bF,  W_{1:k}) =  H(S_{\mathcal{A}}, A^{[k+1]}_{n} |\bF,  W_{1:k}).
\end{align}
\item Invariance: for any $\cA \subset [1:N]$ with $|\cA| = a$ and $k \in [0:K - 1]$ we have
\begin{align}
x_k(a, 0) = H(S_{\mathcal{A}} | W_{1:k})/H(W_1) = y_{k+1}(a, 0). 
\end{align}
\item Boundary: for any disjoint subsets $\cA, \cB \subset [1:N]$, we have
\begin{align}
H(S_{\cA}, A^{[K]}_{\cB} |\bF,  W_{1:N}) = 0.
\end{align}
\end{itemize}
Since any valid symmetrized PIR code must satisfy these constraints, the optimal solution to this LP is indeed a lower bound for $a_0 \alpha+b_0 \beta$ .
\end{proof}

Let us now consider the complexity of this LP. The variables in this LP are all the $x_k(a,b)$'s and $y_k(a,b)$'s for $k=1,2,\ldots,K$, and for integers $a,b$ such that $a+b\leq N$, and it is straightforward to verify that there are a total of $K(N+1) (N+2)$ of them. It is more involved to count the total number of constraints. However, it is clear that the dominant component here is the submodular constraints, and thus let us focus on this set. Let $a, b, c, d, e, f, g, h \in [0 : N]$ be
\begin{align}
&a = |\cA_1 \cap \cA_2|; \quad b = |\cA_1 \cap \cB_2|; \quad c = |\cA_1 \slash (\cA_2 \cup \cB_2)| \notag\\
&d = |\cA_2 \cap \cB_1|; \quad e = |\cB_1 \cap \cB_2|; \quad f = |\cB_1 \slash (\cA_2 \cup \cB_2)| \notag\\
&g = |\cA_2 \slash (\cA_1 \cup \cB_1)|; \quad h = |\cB_2 \slash (\cA_1 \cup \cB_1)|,
\end{align}
then the submodular inequalities can be rewritten as
\begin{align}
&x_k(a + b + c, d + e + f) + x_k(a + g + d, b + e + h) \notag\\
&\qquad\qquad\qquad\geq x_k(a + b + c + d + g, e + f + h) + x_k(a, b + d + e)\\
&y_k(a + b + c, d + e + f) + y_k(a + g + d, b + e + h) \notag\\
&\qquad\qquad\qquad\geq y_k(a + b + c + d + g, e + f + h) + y_k(a, b + d + e). 
\end{align}
Since each of these 8 parameters only takes values in $[0:N]$, the total number of $(a, b, c, d, e, f, g, h)$ combinations is upper-bounded by $(N+1)^8$. 
Moreover, the combinations $(a, b, c, d, e, f, g, h)$ and $(a, d, g, b,$ $e, h, c, f)$ in fact represent the same inequality. 
Thus, there are fewer than $K(N+1)^8$ such submodular inequalities, though the number of valid combinations $(a, b, c, d, e, f, g, h)$ is in fact even smaller due to their inherent relation. Therefore, the problem complexity in terms of the LP constraints is $O(K N^8)$. 

In comparison, let us consider the complexity of the generic entropic LP approach  \cite{yeung1997framework} in the problem setting. The number of random variables there is at least $K+N+KN$, where the $KN$ term is due to the answers from the $N$ servers for the $K$ messages. Thus there are a total of $2^{K+N+KN}-1$ joint entropy values as the variables in the generic entropic LP, and the $K+N+KN+{K+N+KN \choose 2} 2^{K+N+KN-2}$ elemental entropic constraints. It is possible to reduce the number of constraints using the symmetry reduction techniques introduced in \cite{tian2014characterizing,tian2018symmetry}, but it will not change the exponential nature (see \cite{zhang2017symmetry} for a more thorough analysis). In contrast, the complexity of the formulation in Proposition \ref{prop:LP} is polynomial. The significant reduction in the number of constraints is due to the much more restricted set of submodular inequalities we include in this relaxed entropic LP, using the (specific  domain) insights.

\section{Codes to Improve Known Upper Bounds}\label{section-UpperBounds}

The lower convex envelop of the storage-retrieval pairs of the generalized MDS-PIR code provides an upper bound on the optimal tradeoff, which is the best known in the information theoretic PIR formulation. Equipped with the cyclic permutation lemma, in this section, we provide several new base code constructions which yield further improvements. 

\subsection{Construction-A: $N=2$}\label{section-code-KN2}
We provide a construction for $N=2$, which is based on the idea of compressing an existing code \cite{Tian-Sun-Chen-PIR-IT19}. Here the message length $L=1$. We first provide an example, then present the general code construction. 

{\em 
	\noindent\textbf{Example:} Let $(K, N) = (3, 2)$.  There 3 messages  are $a$, $b$ and $c$, respectively. The storage for each server is
	\begin{table}[!h]
		\centering
		\begin{tabular}{|c|c|c|}
			\hline
			$S_1$ & $S_2$ \\ \hline
			$a$ & $a \oplus b$  \\
			$b$ & $a \oplus c$ \\ 
			$c$ &  $\emptyset$ \\ \hline
		\end{tabular}
	\end{table}
	which implies that $\alpha = \frac{5}{2}$. Notice that $b \oplus c$ can be decoded by $S_2$ as $b \oplus c = (a \oplus b) \bigoplus (a \oplus c)$. Suppose the user desires message $a$, it randomly chooses one row to retrieve from the table below as the answers.
	\begin{table}[H]
		\centering
		\begin{tabular}{|c||c|c|c|}
			\hline
			prob. & server 1 & server 2 \\ \hline
			$0.25$ & $a$ & $\emptyset$ \\ \hline
			$0.25$ & $b$ & $a \oplus b$ \\ \hline
			$0.25$ & $c$ & $a \oplus c$ \\ \hline
			$0.25$ & $a \oplus b \oplus c$ & $b \oplus c$ \\ \hline
		\end{tabular}
	\end{table}
	
	Similarly, to retrieve $W_2=b$, the following table is used:
	\begin{table}[H]
		\centering
		\begin{tabular}{|c||c|c|c|}
			\hline
			prob. & server 1 & server 2 \\ \hline
			$0.25$ & $b$ & $\emptyset$ \\ \hline
			$0.25$ & $a$ & $a \oplus b$ \\ \hline
			$0.25$ & $c$ & $b \oplus c$ \\ \hline
			$0.25$ & $a \oplus b \oplus c$ & $a \oplus c$ \\ \hline
		\end{tabular}
	\end{table}

	And to retrieve $W_3 = c$, the user uses the following table:	
	\begin{table}[H]
		\centering
		\begin{tabular}{|c||c|c|}
			\hline
			prob. & server 1 & server 2 \\ \hline
			$0.25$ & $c$ & $\emptyset$ \\ \hline
			$0.25$ & $a$ & $a \oplus c$ \\ \hline
			$0.25$ & $b$ & $b \oplus c$ \\ \hline
			$0.25$ & $a \oplus b \oplus c$ & $a \oplus b$ \\ \hline
		\end{tabular}
	\end{table}
	
	It is seen that $\beta = \frac{7}{8}$. Compared with the capacity achieving code in \cite{sun2019breaking} or \cite{Tian-Sun-Chen-PIR-IT19}, where $(\alpha, \beta) = \left(3, \frac{7}{8} \right)$, the storage is compressed while the download costs remains the same.
}

We next provide the code construction for more general $K$ and $N=2$.
\begin{itemize}
	\item\textbf{Storage: } Let $S_1 = W_{1:K}$, and $S_2 = \{W_1 \oplus W_k\}_{k=2}^K $. It can be interpreted as server $1$ stores all the messages, and server $2$ stores all the even sum of messages, because the summation of any even number of messages can be constructed by $S_2$. The normalized average storage can be calculated simply as 
	\begin{align*}
	\alpha = \frac{1}{2}\left(K + K-1 \right) = K - \frac{1}{2}.
	\end{align*}
	\item \textbf{Retrieval: } To retrieve message $W_k$, $k \in [1:K]$, we randomly choose a length-$K$ vector $\mathbf{v}$ of $0$ and $1$. Then retrieve $ \bigoplus_{ \mathbf{v}(j) = 1} W_j $ and $W_k \oplus \bigoplus_{ \mathbf{v}(j) = 1} W_j$ each bit from one server. If the vector $\mathbf{v}$ has odd number of $1$s, retrieve the former bit from server 2, else retrieve the latter bit from server 2. The user can recover the desired message by $W_k = A_1^{[k]} \oplus A_2^{[k]}$. The user will retrieve 2 bits unless $\mathbf{v}$ is consisted of all $0$s or only the $k^{th}$ position of $\mathbf{v}$ is $1$, and in these two cases, the user only retrieve 1 bit. It follows that
	\begin{align*}
	\beta = \frac{1}{2 \times 2^{K}}(2^K \times 2 - 2 ) = \frac{2^K - 1}{2^K}.
	\end{align*}
	To see that the protocol is private, observe that server $1$ receives queries uniformly distributed over $\{\bigoplus_{\mathbf{v}(j)=1}W_j:$ $\mathbf{v}\text{ has odd number of 1's} \}$; similarly, server $2$ receives queries uniformly distributed over the set $\{\bigoplus_{\mathbf{v}(j)=1}W_j:$ $\mathbf{v}\text{ has even number of 1's} \}$.
\end{itemize}

\subsection{Construction-B: $K|(N-1)$}\label{section-code-Kmod(N-1)}
Next we provide a construction by generalizing code in \cite{sun2019breaking}. Here the message length $L$ is $1$, and $K = T(N-1)$ for some positive integer $T$. An example is given first, and then the general construction will be presented. 

{\em 
	\noindent\textbf{Example:} Let $(K, N) = (4, 3)$, and thus $T = 1$ in this example. There are four messages $a,b,c,d$. The stored contents are as follows
	\begin{table}[!h]
		\centering
		\begin{tabular}{|c|c|c|}
			\hline
			$S_1$ & $S_2$ & $S_3$ \\ \hline
			$a$ & $c$ & $a \oplus c$ \\
			$b$ & $d$ & $b \oplus d$ \\ \hline
		\end{tabular}
	\end{table}
	It is clear that $\alpha = 2$. Suppose the user desires message $a$, a row from the table below is chosen, uniformly at random, as the queries for the servers 
	\begin{table}[!h]
		\centering
		\begin{tabular}{|c||c|c|c|}
			\hline
			prob. & server 1 & server 2 & server 2 \\ \hline
			0.25 & $a$ & $\emptyset$ & $\emptyset$\\ \hline
			0.25 & $b$ & $c \oplus d$ & $a \oplus c \oplus b \oplus d$ \\ \hline
			0.25 & $\emptyset$ & $c$ & $a \oplus c$ \\ \hline
			0.25 & $a \oplus b$ & $d$ & $b \oplus d$ \\ \hline
		\end{tabular}
	\end{table}

	Similarly, to retrieve $W_3 = c$, the following table is used:
	\begin{table}[!h]
		\centering
		\begin{tabular}{|c||c|c|c|}
			\hline
			prob. & server 1 & server 2 & server 2 \\ \hline
			0.25 & $\emptyset$ & $c$ & $\emptyset$\\ \hline
			0.25 & $a \oplus b$ & $d$ & $a \oplus c \oplus b \oplus d$ \\ \hline
			0.25 & $a$ & $\emptyset$ & $a \oplus c$ \\ \hline
			0.25 & $b$ & $c \oplus d$ & $b \oplus d$ \\ \hline
		\end{tabular}
	\end{table}
	It is straightforward to see that code is private because the set of queries for one server is the same and query for retrieving any message is uniformly distributed over that set. Clearly $\beta = 0.75$. Comparing with $(4, 3)$-MDS coded PIR in \cite{Banawan-Ulukus-codedPIR-2018IT}, where $(\alpha, \beta) = (2, 0.802)$, the code given above has a smaller download cost.
}

In the general code construction, the index of a message can be represented either as $(N-1)i + j$, where $i \in [0:T-1]$ and $j \in [1:T]$, or as $Ta + b$, where $a \in [0:N-2]$ and $b \in [1:T]$. 

\begin{itemize}
	\item \textbf{Storage:} For $n\in[1:N-1]$, server $n$ stores $S_n = \{W_{T(n-1) + i}\}_{i = 1}^T$; server $N$ stores $S_N = \{ \bigoplus_{j=0}^{N-2} W_{Tj + i} \}_{i=1}^T$. As a consequence $\alpha=T$.
	
	\item \textbf{Retrieval:} To retrieve $W_{T(a-1) + b}$, where $a \in [1:N-1]$ and $b \in [1:T]$. We randomly generate a vector $\mathbf{v}$ of 0 and 1 with length $T$. Let $\mathbf{v}'$ be a vector such that the only difference between $\mathbf{v}$ is the $b^{th}$ position, which is $\mathbf{v}'(b) = 1 \oplus \mathbf{v}(b)$. Then retrieves $ A_n^{[k]} =\bigoplus_{\mathbf{v}'(i) = 1} W_{T(n-1) + i}$ from server $n \in [1:N-1] \slash \{a\}$; retrieves $ A_a^{[k]} = \bigoplus_{\mathbf{v}(i) = 1} W_{T(a-1) + i}$ from server $a$; and retrieves $ A_N^{[k]} = \bigoplus_{\mathbf{v}'(i) = 1} \bigoplus_{j=0}^{N-2} W_{Tj + i} $ from server $N$. The user can decode $W_k = \bigoplus_{n=1}^N A_n^{[k]}$. Thus
	\begin{align*}
	\beta = \frac{1}{N} \left( \frac{1}{2^T} + \frac{1}{2^T} (N-1) + (1 - \frac{1}{2^{T-1}}) N \right) =\frac{2^T-1}{2^T}.
	\end{align*}
	To see that the protocol is private, observe that for any server $n \in [1:N-1]$, the received query is uniformly distributed over $\left\{ \bigoplus_{\mathbf{v}(i)=1}W_{Ta+i} : \mathbf{v} \in [0:1]^T \right\}$; server $N$ receives queries uniformly distributed over $\Big\{ \bigoplus_{\mathbf{v}(i)=1} \bigoplus_{j=0}^{N-2} W_{Tj + i}:$ $\mathbf{v} \in [0:1]^T \Big\}$.
\end{itemize}

\subsection{Applying the Cyclic Permutation Lemma}\label{section-code-compare-MDSuncoded}
By applying the cyclic permutation lemma to the base codes given as construction-A and construction-B, we can obtain further improvement on the storage-retrieval tradeoff which is given in the following proposition. 
\begin{proposition}\label{prop-new-codes-cyclic}
	For $N,K\geq 2$, the following tradeoff points are achievable:
	\begin{enumerate}[(a)]
		\item $(\alpha,\beta)=\frac{2}{N}\left(K-\frac{1}{2}, \frac{2^K-1}{2^K}\right)$; 
		
		\item $(\alpha,\beta)=\frac{K/T+1}{N}\left(T, \frac{2^T-1}{2^T}\right)$ for all $T$ being a factor of $K$ so that $\frac{K}{T}+1\leq N$. 
	\end{enumerate}
\end{proposition}

\begin{figure}[!th]
	\centering
	\includegraphics[scale=0.6]{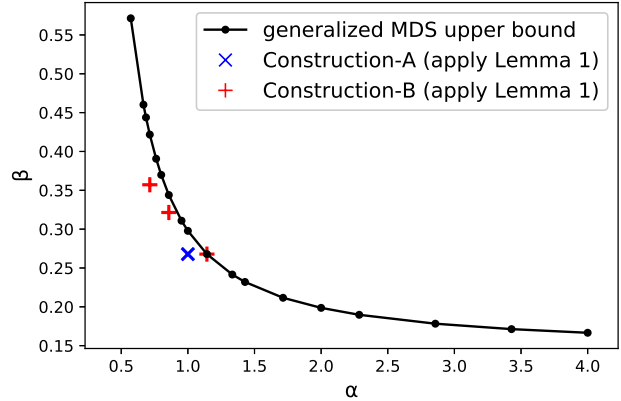}
	\caption{Comparison of new upper bound and generalized MDS upper bound for $(N,K)=(7,4)$.}
	\label{fig_compare-new-gMDS}
\end{figure}

In Fig.~\ref{fig_compare-new-gMDS}, we show the improvement of the tradeoff points in \Cref{prop-new-codes-cyclic} compared with the best known upper bound, i.e. the generalized MDS tradeoff curve. 
Note that there is always a point generated from construction-B lies on the generalized MDS tradeoff curve. 
This can be seen by letting $T=K$ in \Cref{prop-new-codes-cyclic} (b), then the resulting storage-retrieval tradeoff is obtained as
$(\alpha,\beta)=\left(\frac{2K}{N},\frac{1}{N}\frac{2^K-1}{2^{K-1}}\right)$ which is exactly the tradeoff point of the generalized MDS code in \eqref{eqn:gMDS} for $(T_1,T_2)=(1,2)$.

\section{Explicit Lower Bounds}\label{section-LowerBounds}
In this section, we derive more explicit lower bounds by further relaxing the linear program given in \Cref{prop:LP}. The bounds such derived are more explicit, and moreover, we show numerically that the loss from those obtained using Proposition \ref{prop:LP} is small, for the cases that the relaxed entropic LP can be effectively computed. 

\subsection{A Lower Bound for $(N-m)\alpha+m\beta$}\label{section-LB1}

For notation convenience, define the following function for any integer $m\in [1:N]$
\begin{equation}
B_N(K,m)\triangleq \inf_{\text{achievable }(\alpha,\beta)}\big\{(N-m)\alpha+m\beta\big\}.  
\end{equation}
The following boundary conditions are immediate (the last two are from \eqref{def-beta0} and \eqref{def-alpha0}): 
\begin{align}
B_N(1,m)&=1,\text{ for }m\in[1:N],   \label{initial-K=1}  \\
B_N(K,1)&= K,    \label{initial-m=1}  \\
B_N(K,N)&= N \beta_0.   \label{initial-m=N}
\end{align}
Let $\bC$ be the set of positive real-valued vectors $\bm{c}=\left(c_j^n:j\in[1:m],n\in[0:N-j+1]\right)$ satisfying the conditions
\begin{align}
\sum_{n=0}^{N-j+1}c_j^n&=1, \text{ for }j\in[1:m],    \label{coefficient-sumup-to1}  \\
0\leq c_j^n&\leq 1, \text{ for }j\in[1:m], n\in[0:N-j+1],    \label{coefficient-range}
\end{align}
and the conditions 
\begin{align}
\sum_{i=j}^N(N-i+1)d_i\geq \sum_{i=j}^m(N-i+1), \text{ for all }j\in[2:m],  \label{coefficient-SEI-condition}
\end{align} 
where 
\begin{equation}
d_j=
\begin{cases}
\sum_{i=2}^{j-1}c_i^{j-i}\left(\frac{1}{j-i}-\frac{1}{j-1}\right)+c_{j-1}^0+\sum_{n=1}^{N-j+1}\frac{(n-1)c_j^n}{n}, &\text{if }j\in[2:m]\\
c_m^0+\sum_{i=2}^{m}c_i^{j-i}\left(\frac{1}{j-i}-\frac{1}{j-1}\right), &\text{if }j=m+1  \\
\sum_{i=2}^{m}c_i^{j-i}\left(\frac{1}{j-i}-\frac{1}{j-1}\right), &\text{if }j\in[m+2:N]. 
\end{cases}  \label{coefficient-SEI-d}
\end{equation}

For any $\bm{c}\in\bC$, we provide a general lower bound on $B_N(K,m)$ for $m\in[2:N-1]$, through a recursive relation: 
\begin{align}
\widetilde{B_N}(K,m,\bm{c})\triangleq 1+c_1^1(K-1)+\sum_{j=2}^{m}\sum_{n=1}^{N-j+1}\frac{c_j^n}{j+n-1}\widetilde{B_N}(K-1,j+n-1,\bm{c}),   \label{general-LowerBound}
\end{align}
where the initial conditions are given as $\widetilde{B_N}(K,m,\bm{c})=B_N(K,m)$ for i) $K=1$; ii) $m=1$; iii) $m=N$; i.e., the boundary conditions in \eqref{initial-K=1}-\eqref{initial-m=N}. 

Moreover, for $K\geq 2$, $N\geq 2$, and $m\in[1:N]$, define 
\begin{align}
\underline{B_N}(K,m)\triangleq \max_{\bm{c}\in\bC}\widetilde{B_N}(K,m,\bm{c}).
\end{align}

\begin{theorem}\label{thm-general-lower-bound}
	For $K\geq 2$, $N\geq 2$, and $m\in[1:N]$, we have 
	\begin{equation}
	B_N(K,m)\geq \underline{B_N}(K,m) \geq \widetilde{B_N}(K,m,\bm{c}), ~\forall~\bm{c}\in\bC. 
	\end{equation}
\end{theorem}
\begin{proof}[Proof overview for Theorem \ref{thm-general-lower-bound}]
	The lower bound can be obtained by further relaxing the linear program in \Cref{prop:LP}, which is to minimize the objective function under a chosen subset of constraints in a systematic manner. 
	The idea is to first specify which and how the submodularity inequalities are applied (i.e., utilize an even smaller subset of the possible submodular inequalities), and then apply the other inequalities accordingly. 
	More specifically, we use only the submodularity constraints for $j\in[2:N-1],~i\in[j:N]$ that (c.f. \eqref{pf-succ-iteration} and \eqref{lem-geS-pf-1})
	\begin{align}
	&H(A_{1:j}^{[1]}S_{j+1:N}|\bF W_1)+H(S_{j:N}|\bF W_1)\geq H(A_{j}^{[1]}S_{j+1:N}|\bF W_1)+H(A_{1:j-1}^{[1]}S_{j:N}|\bF W_1),  \\
	&H(A_{j:i-1}^{[2]}S_{i:N}|\bF W_1)+H(A_i^{[2]}S_{j:i-1}S_{i+1:N}|\bF W_1)\geq H(A_{j:i}^{[2]}S_{i+1:N}|\bF W_1)+H(S_{j:N}|\bF W_1). 
	\end{align}
	The detailed proof can be found in Appendix~\ref{proof-thm-general-lower-bound}. 
\end{proof}

The bound given in Theorem \ref{thm-general-lower-bound} is still not explicit, and next we specialize it even further, in order to obtain a more explicit form. This is accomplished by choosing a specific set of $\bm{c} \in \bC$. 
More precisely, we will show that the following value of $\dunderline{B_N}(K,m)$, $m\in[2:N-1]$ is feasible, 
\begin{align}
&\dunderline{B_N}(K,m)=   \nonumber \\
&\begin{cases}
1+\frac{m-j^*+1}{N}+\sum_{j=j^*}^{m-1}\frac{(m-j)(m+j-1)}{2j(j-1)N}+\frac{N-j^*}{N-j^*+1}\left(1-\frac{1}{N-j^*}\frac{(m-j^*)(m+j^*-1)}{2(j^*-1)}\right), &\!\!\! \text{if }K=2 \\
1+\sum_{j=j^*}^{m}\frac{1}{j}\cdot \dunderline{B_N}(K-1,j)+\frac{1}{N-j^*+1}\left[\frac{(N-m)(m-1)}{m}-\sum_{i=j^*+1}^m\frac{N-i+1}{i-1}\right] \dunderline{B_N}(K-1,j^*-1), & \!\!\! \text{if }K\geq 3, 
\end{cases}   \label{explicit-LB}
\end{align}
where the initial conditions are given by 
\begin{align}
	\dunderline{B_N}(1,m)&=1,\text{ for }m\in[1:N],     \\
	\dunderline{B_N}(K,1)&= K,     \label{initial-dunderlineB-m=1}   \\
	\dunderline{B_N}(K,N)&= N \beta_0,     \label{initial-dunderlineB-m=N}
\end{align}
and $j^*\in[2:m]$ is defined as follows: 
\begin{itemize}
	\item For $K=2$, $j^*$ is given as
	\begin{equation}
	j^*=\max\left\{2,\left\lceil(N+\frac{1}{2})-\sqrt{(N-m)(N+m-1)+\frac{1}{4}}\right\rceil\right\}.  \label{def-j*-K=2}
	\end{equation}
	
	\item For $K\geq 3$, $j^*$ is the minimum $j$ such that 
	\begin{equation}
	\sum_{i=j+1}^m\frac{N-i+1}{i-1}\leq \frac{(m-1)(N-m)}{m},   \label{def-j*-K>=3}
	\end{equation}
	which is equivalent to 
	\begin{equation}
	N\sum_{i=j}^{m}\frac{1}{i}+j\leq N+1. \label{j*-K>=3}
	\end{equation}
	If the above inequality holds for all $j\in[2:m]$, then let $j^*=2$. 
\end{itemize}

We can upper bound the LHS of \eqref{j*-K>=3} by 
\begin{equation}
N\sum_{i=j}^{m}\frac{1}{i}+j\leq N\cdot\ln\left(\frac{m+1/2}{j-1/2}\right)+j,
\end{equation}
where $f(j)=N\cdot\ln\left(\frac{m+1/2}{j-1/2}\right)+j$ is obtained by the convexity of the reciprocal function $\frac{1}{x}$, which is 
\begin{equation}
\sum_{i=j}^{m}\frac{1}{i}\leq \int_{j-\frac{1}{2}}^{m+\frac{1}{2}}\frac{1}{x}dx=\ln\left(\frac{m+1/2}{j-1/2}\right). 
\end{equation}

The following theorem is our main result of this section.
\begin{theorem}\label{thm-lower-bound-Nm}
	For $K\geq 2$, $N\geq 2$, and $m\in[1:N]$, we have 
	\begin{equation}
	B_N(K,m) \geq \dunderline{B_N}(K,m). 
	\end{equation}
\end{theorem}
\begin{proof}
	The proof can be found in Appendix~\ref{proof-thm-lower-bound-Nm}. 
\end{proof}
\begin{remark}
	The definition of $\dunderline{B_N}(K,m)$ in \eqref{explicit-LB} is defined for $m\in[2:N-1]$. 
	We can verify that for $m=N$, we have $j^*=N$, and then \eqref{explicit-LB} gives $\dunderline{B_N}(K,N)=N\beta_0$ which is consistent with the initial condition in \eqref{initial-dunderlineB-m=N}. However, for $m=1$, we have $j^*=2$, and \eqref{explicit-LB} gives $\dunderline{B_N}(K,1)=1 <=K$, which is not consistent with the initial condition in \eqref{initial-dunderlineB-m=1}. 
\end{remark}

\subsection{Bounding $\alpha+m\beta$ with Large Integer $m$}
Similar to the previous case, we can also relax the relaxed entropic LP to find the following lower bound. 
\begin{theorem}\label{thm-flat-bound}
	For $K\geq 2$, $N\geq 2$, and $m=(N-1)+(N-2)N^{K-k}$ for $k\in[1:K]$, the term $\alpha+m\beta$ can be lower bounded as follows: 
\begin{enumerate}
	\item If $k\leq \frac{K}{2}$, then 
	\begin{align}
	\alpha+m\beta&\geq \left(k+\frac{N^{K-2k}(N^k-1)(N-2)}{(N-1)}\right) +\frac{1}{N^{k-1}}\left[\frac{(N^{K-k}-1)(N-2)}{N(N-1)}+(K-k)\right]  \nonumber \\
	&\quad +\left(1-\frac{1}{N^{k-1}}\right)\Big(B_N(K-k+1,N-1)-1\Big),  \nonumber
	\end{align}
	\item If $k>\frac{K}{2}$, then 
	\begin{align}
	\alpha+m\beta&\geq \left(K-k+\frac{(N-2)(N^{K-k}-1)}{N-1}\right)+ \frac{2(N^{2k-K}-1)}{N^{2k-K}} \nonumber\\
	&\quad +(N-1)\sum_{i=1}^{2k-K-1}\frac{1}{N^i}\Big(B_N(k-i+1,N-1)-1\Big) \nonumber \\
	&\quad +\frac{1}{N^{k-1}}\left[\frac{(N^{K-k}-1)(N-2)}{N(N-1)}+(K-k)\right] +\frac{N^{K-k}-1}{N^{k-1}}\Big(B_N(K-k+1,N-1)-1\Big).  \nonumber
	\end{align}
	\end{enumerate}
\end{theorem}
\begin{proof}
	Similar to \Cref{thm-general-lower-bound}, the lower bound can be obtained by minimizing the objective function in \Cref{prop:LP} under a chosen subset of constraints. 
	Specifically, we use the submodularity constraints for $j\in[2:N-1],~i\in[1:j]$ that (c.f. \eqref{pf-flat-iteration-2} and \eqref{pf-flat-iteration-5})
	\begin{align*}
	&H(A_{1:N}^{[k]}|\bF W_{1:k})+H(S_{j+1:N}|\bF W_{1:k})\geq H(A_{1:j}^{[k]}S_{j+1:N}|\bF W_{1:k})+H(A_{j+1:N}^{[k]}|\bF W_{1:k}),   \\
	&H(A_{1:i}^{[k+1]}S_{j+1:N}|\bF W_{1:k})+H(A_{i+1}^{[k+1]}S_{j+1:N}|\bF W_{1:k}) \geq H(A_{1:i+1}^{[k+1]}S_{j+1:N}|\bF W_{1:k})+H(S_{j+1:N}|\bF W_{1:k}).
	\end{align*}
	The details can be found in Appendix~\ref{proof-thm-flat-bound}. 
\end{proof}
\begin{remark}
	From Fig.~\ref{fig_compare-LP-Tian19}, we see that the two constraints given in~\cite{tian2019storage} are also obtained through the relaxed entropic LP. 
	Now these two constraints are also included in the further relaxed explicit expression, where the first one (Theorem 1 in \cite{tian2019storage}) is simply \eqref{initial-m=1} and the second one (Theorem 2 in \cite{tian2019storage}) can be obtained from \Cref{thm-flat-bound} by letting $k=1$ which becomes 
	\begin{align}
	\alpha+\left[(N-1)+(N-2)N^{K-1}\right]\beta\geq K+\frac{(N-2)(N^K-1)}{N(N-1)}. 
	\end{align}
\end{remark}

\subsection{Numerical Results}\label{section-numerical}
We first compare the proposed lower bounds and the upper bounds of the storage-retrieval tradeoff in Fig.~\ref{fig_lower-upper-bounds}. 
\begin{figure}[t!]
	\centering
	\begin{subfigure}[t]{0.3\textwidth}
		\centering
		\includegraphics[draft=false,scale=0.35]{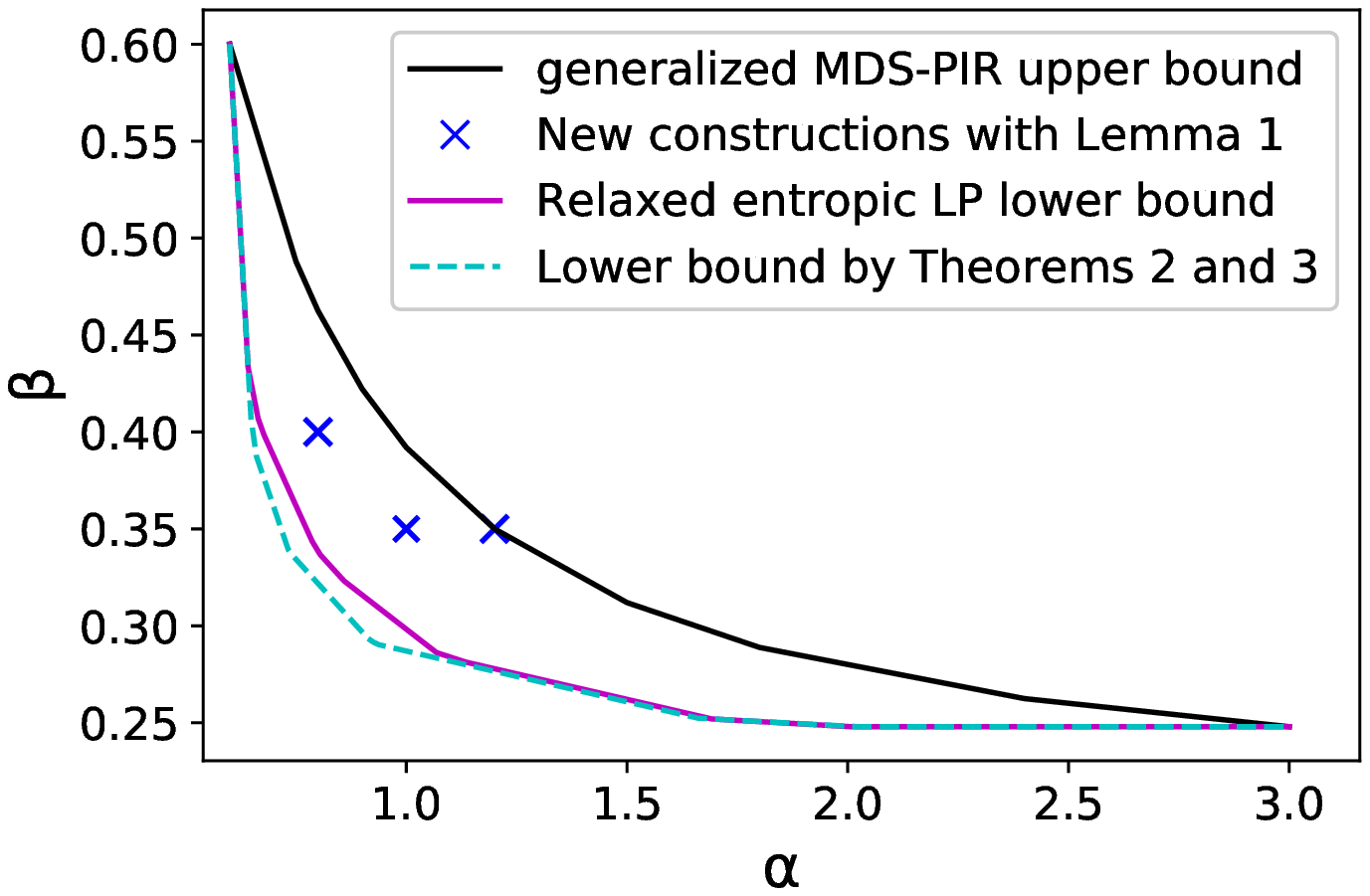}
		\caption{}
	\end{subfigure}
	\begin{subfigure}[t]{0.3\textwidth}
		\centering
		\includegraphics[draft=false,scale=0.35]{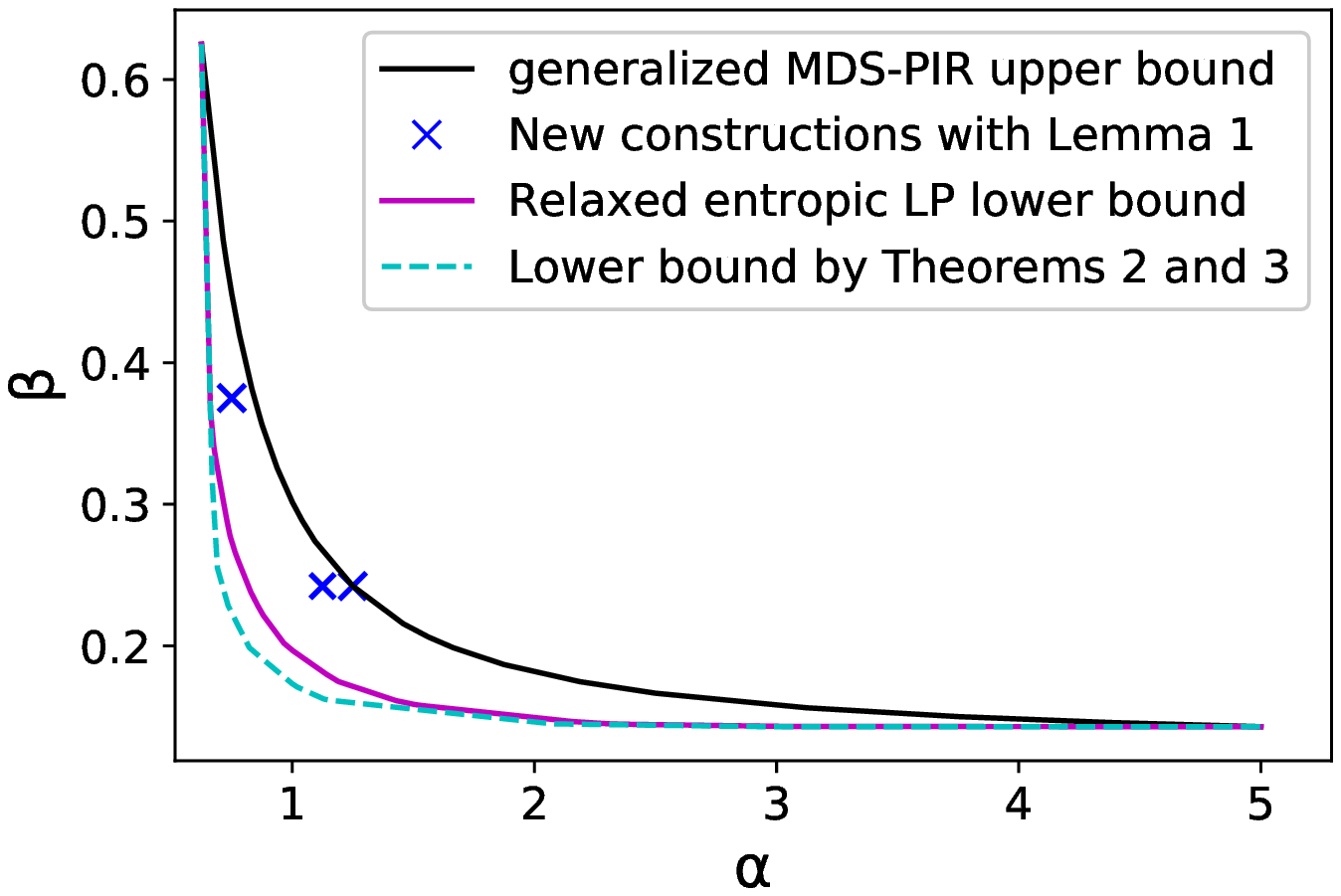}
		\caption{}
	\end{subfigure}
	\begin{subfigure}[t]{0.3\textwidth}
		\centering
		\includegraphics[draft=false,scale=0.35]{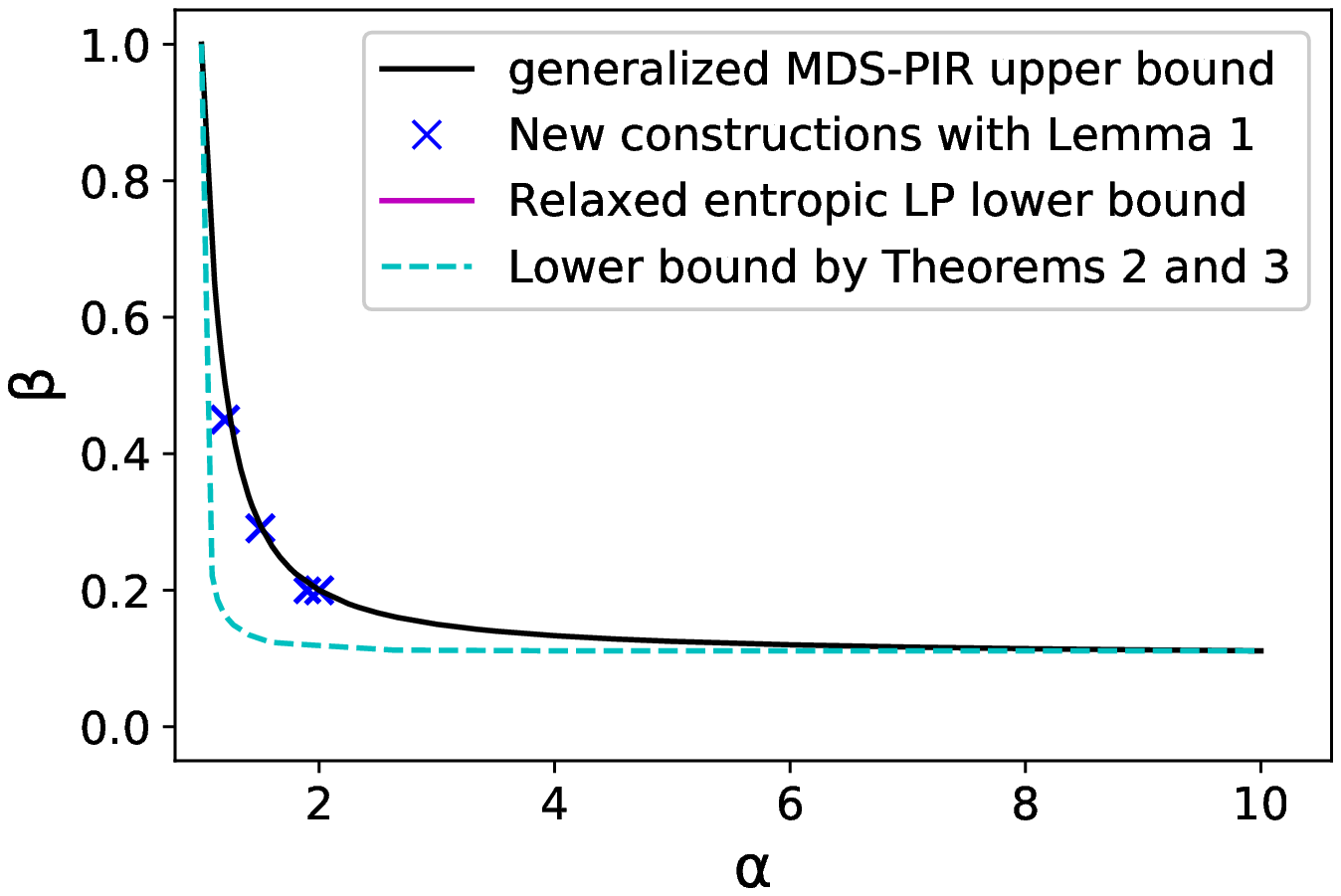}
		\caption{}
	\end{subfigure}
	\caption{Comparison of the lower bounds and upper bounds for $(N,K)=(5,3), (8,5), (10,10)$.}
	\label{fig_lower-upper-bounds}
\end{figure}
The upper bounds in \Cref{section-UpperBounds} evidently outperform the generalized MDS-PIR upper bound at the storage-bound regime. We further observe that the lower bound in \Cref{thm-lower-bound-Nm} and \Cref{thm-flat-bound} is close to the relaxed entropic LP lower bound in \Cref{section-LP-LowerBound}. 
Recall the relaxed entropic LP has constraints on the order of $O(K N^8)$, which becomes unmanageable for larger $N,K$, e.g.,  we found $K=N=10$ can not be effectively computed in a reasonable amount of time. However, the lower bound in \Cref{thm-lower-bound-Nm} and \Cref{thm-flat-bound} is obtained only by direct calculation and can be solved quickly for large $K$ and $N$. 
These upper bounds and lower bounds help to further refine the approximation. 

We next further analyze the difference in Fig.~\ref{fig_gap-lower-upper-bounds}, and consider whether the new bounds will be able to provide a tighter approximation ratio. Since both the lower bounds and the upper bounds are small for large $K$ and $N$, we plot the ratio of the upper and lower bounds instead of their difference. Unfortunately, though the new bounds indeed provide improvement over the existing art, it appears they are not sufficient to yield a better approximation of $\beta$ in the storage-bound regime, and the largest ratio gap appears to be just above $\frac{K}{N}$; the precise positions of the largest gap are given in the respective figures.

\begin{figure}[!t]
	\centering
	\begin{subfigure}[t]{0.3\textwidth}
		\centering
		\includegraphics[draft=false,scale=0.35]{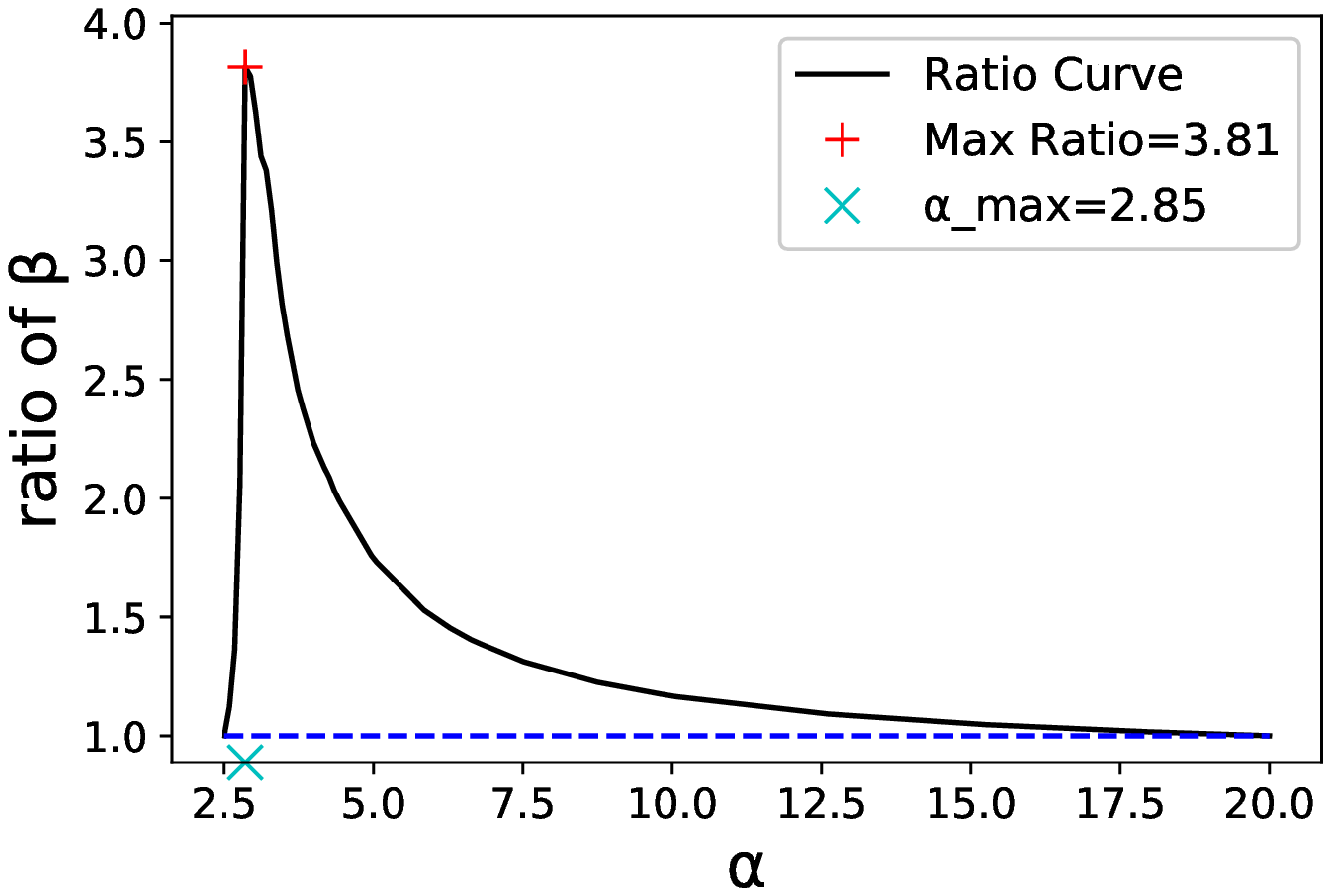}
		\caption{}
	\end{subfigure}
	\begin{subfigure}[t]{0.3\textwidth}
		\centering
		\includegraphics[draft=false,scale=0.35]{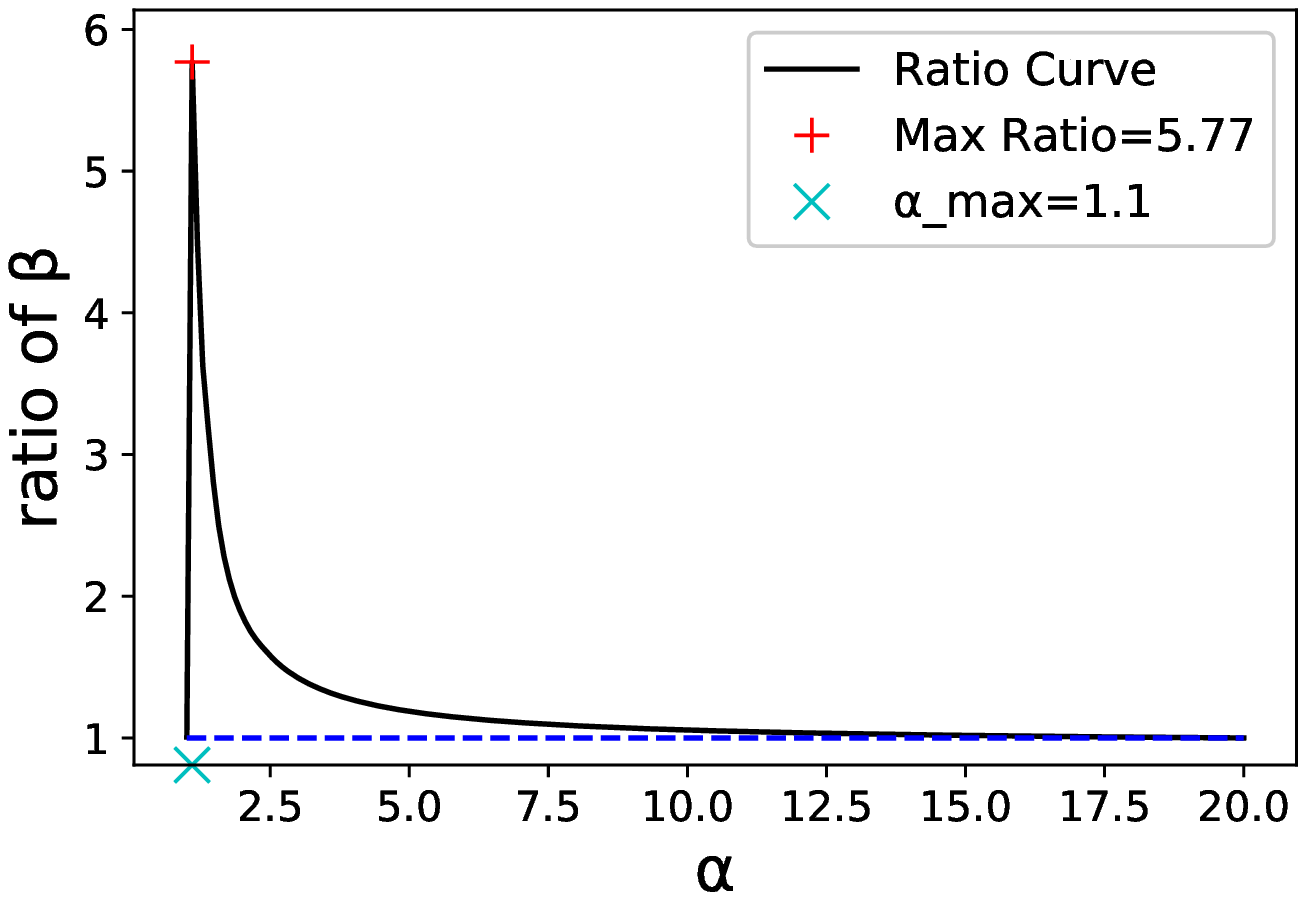}
		\caption{}
	\end{subfigure}
	\begin{subfigure}[t]{0.3\textwidth}
		\centering
		\includegraphics[draft=false,scale=0.35]{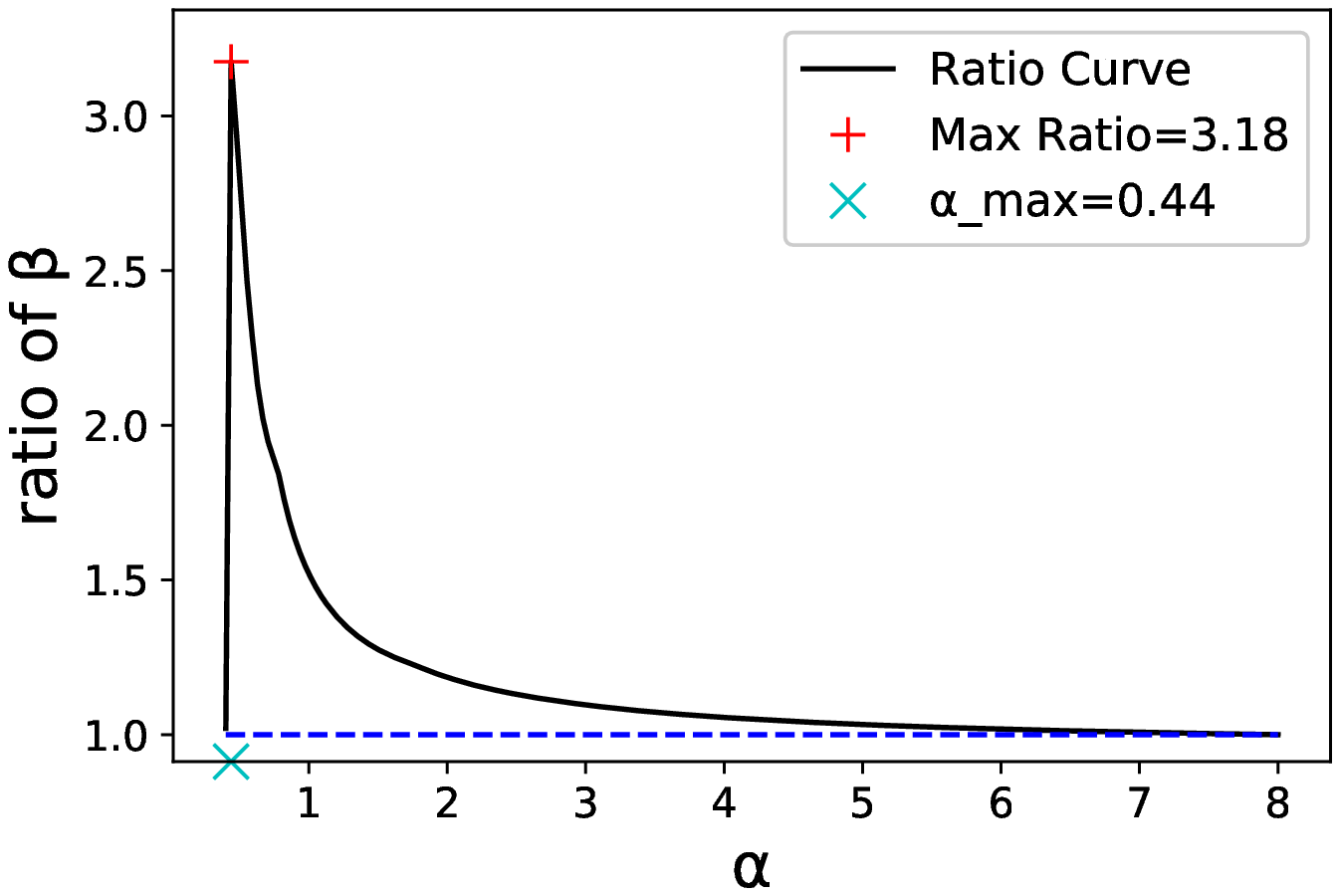}
		\caption{}
	\end{subfigure}
	\caption{The ratio of the upper bound and lower bound for $(N,K)=(8,20), (20,20), (20,8)$.}
	\label{fig_gap-lower-upper-bounds}
\end{figure}

\section{Conclusion}\label{section-conclusion}
We studied the tradeoff between the storage cost and the download cost in private information retrieval systems. Three fundamental results are first presented: a regime-wise 2-approximation, a cyclic permutation lemma, and a relaxed entropic LP with polynomial complexity. Equipped with these results, we then provide improved upper bounds and lower bounds. Though these results provide significant new insights into the storage-retrieval tradeoff in PIR systems, the characterization is not tight in general. As a future work, we plan to further investigate the relaxed entropic LP and derive improved lower bound that can yield better approximate or precise characterizations.

\begin{appendices}
\section{Proof of \Cref{thm-general-lower-bound}}\label{proof-thm-general-lower-bound}
For $m=1$ and $m=N$, we have $B_N(K,1)= K$ and $B_N(K,N)= N\beta_0$ which are the boundary conditions in \eqref{initial-m=1} and \eqref{initial-m=N}. 
For $m\in[2:N-1]$, we prove the lower bound by first considering 
\begin{align}
B_N(K,m)&\geq \frac{1}{L}H(A_{1:m}^{[1]}S_{m+1:N}|\bF)  \\
&=\frac{1}{L}\left[H(W_1)+H(A_{1:m}^{[1]}S_{m+1:N}|\bF W_1)\right], \\
&=1+\frac{1}{L} H(A_{1:m}^{[1]}S_{m+1:N}|\bF W_1).   \label{alpha-to-entropy}
\end{align}
Thus, we only need to find a lower bound to $H(A_{1:m}^{[1]}S_{m+1:N}|\bF W_1)$. 
We start from the following lemmas. 
\begin{lemma}\label{lemma-suss-add}
	For $m\in[2:N-1]$, we have 
	\begin{equation}
	H(A_{1:m}^{[1]}S_{m+1:N}|\bF W_1)+\sum_{j=2}^mH(S_{j:N}|\bF W_1)\geq \sum_{j=1}^mH(A_j^{[2]}S_{j+1:N}|\bF W_1). \label{succ-add}
	\end{equation}
\end{lemma}
\begin{proof}
	For $j\in[2:N-1]$, using submodularity, we have 
	\begin{align}
	H(A_{1:j}^{[1]}S_{j+1:N}|\bF W_1)+H(S_{j:N}|\bF W_1)\geq H(A_{j}^{[1]}S_{j+1:N}|\bF W_1)+H(A_{1:j-1}^{[1]}S_{j:N}|\bF W_1).  \label{pf-succ-iteration}
	\end{align}
	summing up \eqref{pf-succ-iteration} for $j=2,3,\cdots,m$, we obtain 
	\begin{align}
	H(A_{1:m}^{[1]}S_{m+1:N}|\bF W_1)+\sum_{j=2}^mH(S_{j:N}|\bF W_1)\geq \sum_{j=1}^mH(A_j^{[1]}S_{j+1:N}|\bF W_1). 
	\end{align}
	By replacing the terms in the RHS using identical distribution (privacy), we can obtain the inequality in~\eqref{succ-add}. 
\end{proof}

\begin{lemma}\label{lemma-generageS}
	For any $j\in[2:N-1]$ and $n\in[1:N-j+1]$, we have 
	\begin{align}
	H(A^{[2]}_jS_{j+1:N}|\bF W_1)&\geq \frac{1}{j+n-1}H(A^{[2]}_{1:j+n-1}S_{j+n:N}|\bF W_1)   \nonumber \\
	&\qquad +\left(\frac{1}{n}-\frac{1}{j+n-1}\right)H(S_{j+n:N}|\bF W_1)+\frac{n-1}{n}H(S_{j:N}|\bF W_1).    \label{generateS-1}
	\end{align}
	For $j=1$, we have 
	\begin{equation}
	H(A^{[2]}_1S_{2:N}|\bF W_1)\geq (K-1)L.   \label{generateS-2}
	\end{equation}
	For $j\in[1:N-1]$, we have 
	\begin{equation}
	H(A^{[2]}_jS_{j+1:N}|\bF W_1)\geq H(S_{j+1:N}|\bF W_1).     \label{generateS-3}
	\end{equation}
\end{lemma}
\begin{proof}
	The inequality in \eqref{generateS-3} is the monotonicity in \eqref{entropicLP-monotone}. 
	The inequality in \eqref{generateS-2} follows by applying the following inequality successively for $k=1,2,\cdots,K-1$: 
	\begin{align}
	H(A^{[k+1]}_1S_{2:N}|\bF W_{1:k})&=H(A^{[k+1]}_1S_{2:N}W_{k+1}|\bF W_{1:k})   \label{pf-generateS-3-1}  \\
	&=L+H(A^{[k+1]}_1S_{2:N}|\bF W_{1:k+1})  \label{pf-generateS-3-2}  \\
	&=L+H(A^{[k+1]}_1S_{2:N}|Q_n^{[k+1]} W_{1:k+1})   \label{pf-generateS-3-3}  \\
	&=L+H(A^{[k+2]}_1S_{2:N}|Q_n^{[k+2]} W_{1:k+1})  \label{pf-generateS-3-4}  \\
	&=L+H(A^{[k+2]}_1S_{2:N}|\bF W_{1:k+1}),  \label{pf-generateS-3-5}  
	\end{align}
	where \eqref{pf-generateS-3-3} and \eqref{pf-generateS-3-5} follow from the Markov chains $\bF\leftrightarrow Q_n^{[k+1]}\leftrightarrow (W_{1:K}S_{1:N}A_n^{[k+1]})$ and $\bF\leftrightarrow Q_n^{[k+2]}\leftrightarrow (W_{1:K}S_{1:N}A_n^{[k+2]})$,
	and \eqref{pf-generateS-3-4} follows from the privacy requirement. 
	To prove \eqref{generateS-1}, we first consider for any $j\in[2:N-1]$ and $i\in[j:N]$ that 
	\begin{equation}
	H(A_{j:i-1}^{[2]}S_{i:N}|\bF W_1)+H(A_i^{[2]}S_{j:i-1}S_{i+1:N}|\bF W_1)\geq H(A_{j:i}^{[2]}S_{i+1:N}|\bF W_1)+H(S_{j:N}|\bF W_1),  \label{lem-geS-pf-1}
	\end{equation}
	which can be obtained directly by applying submodularity. 
	Then by symmetry, for  $j\in[2:N-1]$ and $n\in[1:N-j+1]$, we have 
	\begin{align}
	H(A^{[2]}_jS_{j+1:N}|\bF W_1)&\geq \frac{1}{n}\sum_{i=j}^{j+n-1}H(A^{[2]}_iS_{j:i-1,i+1:N}|\bF W_1)   \\
	&\geq \frac{1}{n}\left[H(A^{[2]}_{j:j+n-1}S_{j+n:N}|\bF W_1)+(n-1)H(S_{j:N}|\bF W_1)\right]   \label{lem-geS-pf-2} \\
	&\geq \frac{1}{j+n-1}H(A^{[2]}_{1:j+n-1}S_{j+n:N}|\bF W_1)   \nonumber \\
	&\qquad +\left(\frac{1}{n}-\frac{1}{j+n-1}\right)H(S_{j+n:N}|\bF W_1)+\frac{n-1}{n}H(S_{j:N}|\bF W_1),    \label{lem-geS-pf-3}
	\end{align}
	where \eqref{lem-geS-pf-2} is obtained by applying \eqref{lem-geS-pf-1} for $i=j+1,j+2\cdots,j+n-1$, 
	and the last inequality \eqref{lem-geS-pf-3} follows from symmetry and the Han's inequality in \eqref{entropicLP-Han} that 
	\begin{align}
	\frac{1}{n}H(A^{[2]}_{j:j+n-1}|S_{j+n:N}\bF W_1)&=\frac{1}{{j+n-1 \choose n}}\sum_{\cB\subseteq[1:j+n-1]:|\cB|=n}\frac{H(A^{[2]}_{\cB}|S_{j+n:N}\bF W_1)}{n}   \\
	&\geq \frac{1}{j+n-1}H(A^{[2]}_{1:j+n-1}|S_{j+n:N}\bF W_1). 
	\end{align}
\end{proof}

\begin{lemma}\label{lemma-LB-reduce}
	For $K\geq 3$, $N\geq 3$, $m\in[1:N-1]$, and $k\in[0:K-1]$, we have 
	\begin{equation}
	\frac{1}{L}H(A_{1:m}^{[k+1]}S_{m+1:N}|FW_{1:k})\geq B_N(K-k,m). 
	\end{equation}
\end{lemma}
\begin{proof}
	The lemma can be obtained if we view the messages $W_{1:k}$ as empty sets and then there are $K-k$ messages in the system. 
\end{proof}

The following {\it subset entropy inequality} will be used in the bounding process. 
\begin{lemma}\label{lemma-SubsetEntropyInequality}
	For $d_j\geq 0$, $j\in[2:N]$, if 
	\begin{equation}
	\sum_{i=j}^N(N-i+1)d_i\geq \sum_{i=j}^m(N-i+1), \text{ for all }j\in[2:m],  \label{SEI-condition}
	\end{equation}
	then we have 
	\begin{equation}
	\sum_{j=2}^Nd_j H(S_{j:N})\geq \sum_{j=2}^m H(S_{j:N}). \label{SEI-inequality}
	\end{equation}
\end{lemma}
\begin{proof}
	The lemma can be obtained by symmetry and Han's inequality. The details can be found in Appendix~\ref{proof-lemma-SubsetEntropyInequality}. 
\end{proof}

Now we use \Cref{lemma-generageS} to bound each term in the RHS of \eqref{succ-add}. 
We partition each term into several pieces, each of which applies one of the inequalities in \eqref{generateS-1}-\eqref{generateS-3}. 
For $j\in[1:m]$, denote the coefficient of the inequality in \eqref{generateS-3} by $c_j^0$. 
For $j=1$, let the coefficient of the inequality in \eqref{generateS-2} be $c_1^1$, and for notational simplicity in the sequel, let $c_1^n=0$ for $n\in[2:N-1]$. 
For $j\in[2:m]$, denote the coefficient of the inequality in \eqref{generateS-1} by $c_j^n$ for $n\in[1:N-j+1]$. 
Then by \Cref{lemma-suss-add}, we have the following bound, 
\begin{align}
&H(A_{1:m}^{[1]}S_{m+1:N}|\bF W_1)   \nonumber \\
&\geq \sum_{j=1}^mH(A_j^{[2]}S_{j+1:N}|\bF W_1)-\sum_{j=2}^mH(S_{j:N}|\bF W_1)   \label{LB-entropy-1}  \\
&\geq \sum_{j=1}^{m}c_j^0H(S_{j+1:N}|\bF W_1)+c_1^1(K-1)L  \nonumber \\
&\quad +\sum_{j=2}^{m}\sum_{n=1}^{N-j+1}c_j^n\left[\frac{1}{j+n-1}H(A^{[2]}_{1:j+n-1}S_{j+n:N}|\bF W_1)\right.   \nonumber \\
&\qquad\quad +\left.\left(\frac{1}{n}-\frac{1}{j+n-1}\right)H(S_{j+n:N}|\bF W_1)+\frac{n-1}{n}H(S_{j:N}|\bF W_1)\right]-\sum_{j=2}^mH(S_{j:N}|\bF W_1)    \label{LB-entropy-2} \\
&=c_1^1(K-1)L+\sum_{j=2}^{m}\sum_{n=1}^{N-j+1}\frac{c_j^nL}{j+n-1}{B_N}(K-1,j+n-1)+\sum_{j=2}^{m+1}c_{j-1}^0H(S_{j:N}|\bF W_1)  \nonumber \\
&\quad +\sum_{j=2}^{m}\sum_{n=1}^{N-j+1}c_j^n\left[\left(\frac{1}{n}-\frac{1}{j+n-1}\right)H(S_{j+n:N}|\bF W_1)+\frac{n-1}{n}H(S_{j:N}|\bF W_1)\right]-\sum_{j=2}^mH(S_{j:N}|\bF W_1)   \label{LB-entropy-3}\\
&=c_1^1(K-1)L+\sum_{j=2}^{m}\sum_{n=1}^{N-j+1}\frac{c_j^nL}{j+n-1}{B_N}(K-1,j+n-1)   \nonumber \\
&\quad +\sum_{j=2}^m\left[\sum_{i=2}^{j-1}c_i^{j-i}\left(\frac{1}{j-i}-\frac{1}{j-1}\right)+c_{j-1}^0+\sum_{n=1}^{N-j+1}\frac{(n-1)c_j^n}{n}\right]H(S_{j:N}|\bF W_1)+c_m^0H(S_{m+1:N}|\bF W_1)   \nonumber \\
&\quad +\sum_{j=m+1}^{N}\left[\sum_{i=2}^{m}c_i^{j-i}\left(\frac{1}{j-i}-\frac{1}{j-1}\right)\right]H(S_{j:N}|\bF W_1)-\sum_{j=2}^mH(S_{j:N}|\bF W_1)  \label{LB-entropy-4}\\
&\geq c_1^1(K-1)L+\sum_{j=2}^{m}\sum_{n=1}^{N-j+1}\frac{c_j^nL}{j+n-1}{B_N}(K-1,j+n-1),   \label{LB-entropy}
\end{align}
where \eqref{LB-entropy-2} follows from \Cref{lemma-generageS}, 
\eqref{LB-entropy-3} follows from \Cref{lemma-LB-reduce}, 
\eqref{LB-entropy-4} is a reorganization of the terms, 
and the last inequality \eqref{LB-entropy} follows from a conditional version of \Cref{lemma-SubsetEntropyInequality} and the assumption of the condition 
\begin{equation}
\sum_{i=j}^N(N-i+1)d_i\geq \sum_{i=j}^m(N-i+1), \text{ for all }j\in[2:m],  
\end{equation} 
holds for 
\begin{equation}
d_j=
\begin{cases}
\sum_{i=2}^{j-1}c_i^{j-i}\left(\frac{1}{j-i}-\frac{1}{j-1}\right)+c_{j-1}^0+\sum_{n=1}^{N-j+1}\frac{(n-1)c_j^n}{n}, &\text{if }j\in[2:m]\\
c_m^0+\sum_{i=2}^{m}c_i^{j-i}\left(\frac{1}{j-i}-\frac{1}{j-1}\right), &\text{if }j=m+1  \\
\sum_{i=2}^{m}c_i^{j-i}\left(\frac{1}{j-i}-\frac{1}{j-1}\right), &\text{if }j\in[m+2:N]. 
\end{cases}  
\end{equation}
By substituting \eqref{LB-entropy} into \eqref{alpha-to-entropy}, we obtain for $m\in[2:N-1]$ that 
\begin{equation}
B_N(K,m)\geq 1+c_1^1(K-1)+ \sum_{j=2}^{m} \sum_{n=1}^{N-j+1} \frac{c_j^n}{j+n-1} {B_N}(K-1,j+n-1).   \label{recursive-inequality-B}
\end{equation}
The above inequality regarding $B_N(K,m)$ is recursive on $K$. If we replace the inequalities with equalities for all recursions, the resulting objective value should be a lower bound of $B_N(K,m)$, which is exactly the recursive definition of $\widetilde{B_N}(K,m)$ in~\eqref{general-LowerBound}. This proves the theorem.

\section{Proof of \Cref{lemma-SubsetEntropyInequality}}\label{proof-lemma-SubsetEntropyInequality}
By database symmetry, we have 
\begin{equation}
H(S_{\cB_1})=H(S_{\cB_2}),~\forall \cB_1,\cB_2\subseteq[1:N], |\cB_1|=|\cB_2|. 
\end{equation}
Then Han's inequality becomes 
\begin{equation}
\frac{H(S_{i:N})}{N-i+1}\geq \frac{H(S_{j:N})}{N-j+1},~\forall i,j\in[1:N], i\geq j.   \label{pf-SEI-Han}
\end{equation}
Consider $d_2,d_3,\cdots,d_N$ satisfying \eqref{SEI-condition}, i.e., 
\begin{equation}
\sum_{i=j}^N(N-i+1)d_i\geq \sum_{i=j}^m(N-i+1), ~\forall j\in[2:m]. 
\end{equation}
We can partition each $d_i$ into $d_i^2,d_i^3,\cdots,d_i^{m}$ so that 
\begin{equation}
\sum_{i=j}^N(N-i+1)d_i^j\geq (N-j+1), ~\forall j\in[2:m].   \label{SEI-condition-partition}
\end{equation}
For $ j\in[2:m]$, we have 
\begin{align}
\sum_{i=j}^Nd_i^jH(S_{i:N})\geq \sum_{i=j}^N\frac{(N-i+1)d_i^j}{N-j+1}H(S_{j:N}) \geq H(S_{j:N}),
\end{align}
where the two inequalities follow from \eqref{pf-SEI-Han} and \eqref{SEI-condition-partition}, respectively.
Summing up the above inequality over $j\in[2:m]$, we obtain 
\begin{align}
\sum_{j=2}^m H(S_{j:N})\leq \sum_{j=2}^m\sum_{i=j}^Nd_i^jH(S_{i:N})= \sum_{i=2}^N\left(\sum_{j=2}^i d_i^j\right)H(S_{i:N}) \leq \sum_{i=2}^N\left(\sum_{j=2}^m d_i^j\right)H(S_{i:N})= \sum_{i=2}^Nd_jH(S_{i:N}), 
\end{align}
which is the inequality in \eqref{SEI-inequality}. This proves the lemma.

\section{Proof of \Cref{thm-lower-bound-Nm}}\label{proof-thm-lower-bound-Nm}
We only need to prove for $m\in[2:N-1]$ that $\dunderline{B}_N(K,m)=\widetilde{B_N}(K,m,\bm{c})$ for some $\bm{c}\in\bC$. 
For $K=2$, we assign the following values 
\begin{equation}
c_j^n=
\begin{cases}
\frac{(m-j)(m+j-1)}{2(j-1)(N-j)},& \text{if } j\geq j^* \text{ and } n=1 \\
1-\frac{(m-j)(m+j-1)}{2(j-1)(N-j)}, &\text{if } j\geq j^* \text{ and } n=N-j+1 \\
1-\frac{(N-j^*)(j^*-1)}{N-j^*+1}\left(1-\frac{(m-j^*)(m+j^*-1)}{2(j^*-1)(N-j^*)}\right), &\text{if } j=j^*-1 \text{ and } n=0 \\
\frac{(N-j^*)(j^*-1)}{N-j^*+1}\left(1-\frac{(m-j^*)(m+j^*-1)}{2(j^*-1)(N-j^*)}\right),& \text{if } j=j^*-1 \text{ and } n=1 \\
1,&\text{if }j< j^*-1 \text{ and } n=0 \\
0, & \text{ otherwise,}
\end{cases}  \label{c_j^n-assign-K=2}
\end{equation}
which is easily seen satisfies \eqref{coefficient-sumup-to1}. 
Substituting the above value into \eqref{coefficient-SEI-d}, we have $d_j=1$ for $j\in[2:m]$ and $d_j=0$ for $j\notin[j^*:m]$. 
It is obvious that \eqref{coefficient-SEI-condition} is satisfied for all $j\in[2:m]$. 
To show the above coefficient $c_j^n$ satisfies \eqref{coefficient-range}, we only need to prove 
\begin{align}
\frac{(m-j)(m+j-1)}{2(j-1)(N-j)}\leq 1,~\forall j\geq j^*  \label{pf-coefficient-range-require1}
\end{align}
and
\begin{align}
\frac{(N-j^*)(j^*-1)}{N-j^*+1}\left(1-\frac{(m-j^*)(m+j^*-1)}{2(j^*-1)(N-j^*)}\right)\leq 1. \label{pf-coefficient-range-require2}
\end{align}
Let $f_1(j)=2(j-1)(N-j)-(m-j)(m+j-1)=\left(N-\frac{1}{2}\right)^2-\left(N+\frac{1}{2}-j\right)^2-m^2+m$ which is an increasing function of $j$. 
Let $f_2(j)=f_1(j)-2(N-j+1)$. 
\begin{itemize}
	\item If $j^*=2$, from the definition of $j^*$ in \eqref{def-j*-K=2}, we have $2\geq \left\lceil(N+\frac{1}{2})-\sqrt{(N-m)(N+m-1)+\frac{1}{4}}\right\rceil$, i.e., $2(N-1)\geq m^2-m$. 
	Then $f_1(j)\geq f_1(2)=2(N-1)-m^2+m\geq 0$ for any $j\geq j^*$, which proves \eqref{pf-coefficient-range-require1}. 
	Moreover, we have $f_2(j^*)=f_1(2)-2(N-1)=-m(m-1)<0$, which proves \eqref{pf-coefficient-range-require2}. 
	
	\item If $j^*=\left\lceil(N+\frac{1}{2})-\sqrt{(N-m)(N+m-1)+\frac{1}{4}}\right\rceil$, from \eqref{def-j*-K=2}, we have $2(N-1)< m^2-m$. 
	Then for any $j\geq j^*$, we obtain
	\begin{align}
	f_1(j)&\geq f_1(j^*)  \\
	&\geq \left(N-\frac{1}{2}\right)^2-\left(N+\frac{1}{2}-(N+\frac{1}{2})+\sqrt{(N-m)(N+m-1)+\frac{1}{4}}\right)^2-m^2+m  \\
	&=\left(N-\frac{1}{2}\right)^2-(N-m)(N+m-1)-\frac{1}{4}-m^2+m  \\
	&= 0,
	\end{align}
	which proves \eqref{pf-coefficient-range-require1}. We can then further obtain 
	\begin{align}
	f_2(j^*)&=f_1(j^*)-2(N-j^*+1) \\
	&\leq \left(N-\frac{1}{2}\right)^2-\left(N+\frac{1}{2}-(N+\frac{1}{2})+\sqrt{(N-m)(N+m-1)+\frac{1}{4}}-1\right)^2-m^2+m  \nonumber \\
	&\qquad -2(N-j^*+1)  \\
	&=\left(N-\frac{1}{2}\right)^2-\left(\sqrt{(N-m)(N+m-1)+\frac{1}{4}}-1\right)^2-m^2+m -2(N-j^*+1)  \\
	&=2\sqrt{(N-m)(N+m-1)+\frac{1}{4}}-1-2(N-j^*+1)  \\
	&\leq 2\sqrt{(N-m)(N+m-1)+\frac{1}{4}}-2N-3  \nonumber \\
	&\quad +2\left((N+\frac{1}{2})-\sqrt{(N-m)(N+m-1)+\frac{1}{4}}+1\right) \\
	&=0,
	\end{align}
	which implies that \eqref{pf-coefficient-range-require2} is true. 
\end{itemize}

For $K\geq 3$ and $j\geq j^*$, we assign the following values 
\begin{equation}
c_j^n=
\begin{cases}
1,& \text{if } j\geq j^* \text{ and } n=1 \\
1-\frac{(N-j^*)(j^*-1)^2}{j^*(N-j^*+1)}y_{j^*}, &\text{if } j=j^*-1 \text{ and } n=0 \\
\frac{(N-j^*)(j^*-1)^2}{j^*(N-j^*+1)}y_{j^*},& \text{if } j=j^*-1 \text{ and } n=1 \\
1,&\text{if }j< j^*-1 \text{ and } n=0 \\
0, & \text{otherwise,}
\end{cases}
\end{equation}
where $y_{j^*}=\frac{j^*}{(j^*-1)(N-j^*)}\left[\frac{(N-m)(m-1)}{m}-\sum_{i=j^*+1}^m\frac{N-i+1}{i-1}\right]$ is obtained by solving 
\begin{equation}
\begin{cases}
x_j+y_j=1\\
\frac{j-1}{j}x_j+\frac{j(N-j-1)}{(j+1)(N-j)}y_{j+1}=1
\end{cases}  \label{c_j^n-assign-K>=3}
\end{equation}
with initial conditions $x_m=0,y_m=1$. 
We can verify that 
\begin{equation}
d_j=
\begin{cases}
1-\frac{(j^*-1)(N-j^*)}{j^*(N-j^*+1)}y_{j^*},&\text{if }j=j^* \\
\frac{j-2}{j-1},&\text{if }j\in[j^*+1:m+1] \\
0,&\text{o.w.,}
\end{cases}
\end{equation}
and the conditions in \eqref{coefficient-sumup-to1} and \eqref{coefficient-SEI-condition} are easily verified. 
In particular, the equality in \eqref{coefficient-SEI-condition} holds for $j\in[2:j^*]$.
To show the coefficient $c_j^n$ in \eqref{c_j^n-assign-K>=3} satisfies \eqref{coefficient-range}, we only need to prove $\frac{(N-j^*)(j^*-1)^2}{j^*(N-j^*+1)}y_{j^*}\leq 1$, which is equivalent to $\frac{j^*}{N-j^*+1}\left[\frac{(N-m)(m-1)}{m}-\sum_{i=j^*+1}^m\frac{N-i+1}{i-1}\right]\leq 1$. Consider the following two cases: 
\begin{itemize}
	\item If $j^*=2$, we have $\frac{(N-j^*)(j^*-1)^2}{j^*(N-j^*+1)}y_{j^*}=\frac{1}{N-1}\left[\frac{(N-m)(m-1)}{m}-\sum_{i=3}^m\frac{N-i+1}{i-1}\right]< 1$, because $\frac{(N-m)(m-1)}{m}<N-m<N-1$. 
	
	\item If $j^*>2$, by the definition of $j^*$ in \eqref{def-j*-K>=3}, we have $\frac{(N-m)(m-1)}{m}-\sum_{i=j^*+1}^m\frac{N-i+1}{i-1}<\frac{N-j^*}{j^*}$, which implies 
	$\frac{(N-j^*)(j^*-1)^2}{j^*(N-j^*+1)}y_{j^*}<\frac{j^*-1}{N-j^*+1}\cdot \frac{N-j^*}{j^*}<1$. 
\end{itemize}
Now, we have shown that the coefficients designed in both \eqref{c_j^n-assign-K=2} and \eqref{c_j^n-assign-K>=3} satisfy $\bm{c}\in\bC$. 
The value of $\dunderline{B_N}(K,m)$ in \eqref{explicit-LB} is then obtained by substituting the above feasible $\bm{c}$ into the recursive function $\widetilde{B_N}(K,m,\bm{c})$ in \eqref{general-LowerBound}. 
This proves the theorem.

\section{Proof of \Cref{thm-flat-bound}}\label{proof-thm-flat-bound}
To prove the theorem, we begin with the following iterative lemma, for which the proof is given in Appendix~\ref{proof-lemma-iteration-2term}. 
\begin{lemma}\label{lemma-iteration-2term}
	For any $k\in[1:K-1]$, $j\in[1:N-1]$, and any non-negative integer $q$, we have 
	\begin{align}
	&H(A_{1:j}^{[k]}S_{j+1:N}|\bF W_{1:k})+N^q(j-1)H(A_{1:N}^{[k]}|\bF W_{1:k})  \nonumber \\
	&\geq L\cdot\left(1+N^{q-1}(j-1)\right)+\left[H(A_{1:j}^{[k+1]}S_{j+1:N}|\bF W_{1:k+1})+N^{q-1}(j-1)H(A_{1:N}^{[k+1]}|\bF W_{1:k+1})\right].
	\end{align}
\end{lemma}
To simplify the bounding process, we also need the following lemma. 
The lemma can be proved by applying \Cref{lemma-iteration-2term} with $j=N-1$ for $k,k+1,\cdots,K-1$ successively. We omit the details of the proof. 
\begin{lemma}\label{lemma-cancle-entropy}
	For $k\in[1:K-1]$, we have 
	\begin{equation}
	H(A_{1:N-1}^{[k]}S_N|\bF W_{1:k})+N^{K-k-1}(N-2)H(A_{1:N}^{[k]}|\bF W_{1:k})\geq L\cdot\left[\frac{(N^{K-k}-1)(N-2)}{N(N-1)}+(K-k)\right]. 
	\end{equation}
\end{lemma}

Let $\bar{k}=K-k$. Next, we prove the lower bound to $\alpha+m\beta$ for $m=(N-1)+(N-2)N^{K-k}$, $k\in[1:K]$ by successively applying the iterations in \Cref{lemma-iteration-2term}. 
Since the parameter $q$ is non-negative, the lower bound depends on the value of $k$ and is obtained respectively for the following two cases. 

\textit{Case i:} if $k\leq \bar{k}$, i.e., $k\leq \frac{K}{2}$, we apply \Cref{lemma-iteration-2term} till $k$ and obtain 
\begin{align}
&\alpha+\left[(N-1)+(N-2)N^{K-k}\right]\beta   \nonumber \\
&\geq \frac{1}{L}\left[H(A_{1:N-1}^{[1]}S_N|\bF)+N^{K-k-1}(N-2)H(A_{1:N}^{[1]}|\bF)\right]   \\
&\geq \sum_{i=1}^{k}\left[1+(N-2)N^{K-k-i}\right]+\frac{1}{L}\left[H(A_{1:N-1}^{[k]}S_N|\bF W_{1:k})+N^{K-2k}(N-2)H(A_{1:N}^{[k]}|\bF W_{1:k})\right]    \label{flat-bound-case1-1} \\
&= \sum_{i=1}^{k}\left[1+(N-2)N^{K-k-i}\right]+\frac{1}{N^{k-1}L}\left[H(A_{1:N-1}^{[k]}S_N|\bF W_{1:k})+N^{K-k-1}(N-2)H(A_{1:N}^{[k]}|\bF W_{1:k})\right]  \nonumber  \\
&\qquad +\left(1-\frac{1}{N^{k-1}}\right)\frac{1}{L}H(A_{1:N-1}^{[k]}S_N|\bF W_{1:k})    \label{flat-bound-case1-2} \\
&\geq \sum_{i=1}^{k}\left[1+(N-2)N^{K-k-i}\right]+\frac{1}{N^{k-1}}\left[\frac{(N^{K-k}-1)(N-2)}{N(N-1)}+(K-k)\right]  \nonumber  \\
&\qquad +\left(1-\frac{1}{N^{k-1}}\right)\Big(B_N(K-k+1,N-1)-1\Big)    \label{flat-bound-case1-3} \\
&= \left(k+\frac{N^{K-2k}(N^k-1)(N-2)}{(N-1)}\right) +\frac{1}{N^{k-1}}\left[\frac{(N^{K-k}-1)(N-2)}{N(N-1)}+(K-k)\right]  \nonumber \\
&\qquad +\left(1-\frac{1}{N^{k-1}}\right)\Big(B_N(K-k+1,N-1)-1\Big),  \label{flat-bound-case1}
\end{align}
where \eqref{flat-bound-case1-1} follows by applying \Cref{lemma-iteration-2term} with $j=N-1$ for $1,2,\cdots,k-1$ successively, 
and \eqref{flat-bound-case1-3} follows from \Cref{lemma-cancle-entropy}. 

\textit{Case ii:} if $k> \bar{k}$, i.e., $k>\frac{K}{2}$, we apply \Cref{lemma-iteration-2term} till $\bar{k}$ and obtain 
\begin{align}
&\alpha+\left[(N-1)+(N-2)N^{K-k}\right]\beta    \nonumber \\
&\geq \frac{1}{L}\left[H(A_{1:N-1}^{[1]}S_N|\bF)+N^{K-k-1}(N-2)H(A_{1:N}^{[1]}|\bF)\right]   \\
&\geq \sum_{i=1}^{\bar{k}}\left[1+(N-2)N^{K-k-i}\right]+\frac{1}{L}\left[H(A_{1:N-1}^{[\bar{k}]}S_N|\bF W_{1:\bar{k}})+(N-2)H(A_{1:N}^{[\bar{k}]}|\bF W_{1:\bar{k}}) \right]  \label{flat-bound-case2-1} \\
&\geq \sum_{i=1}^{\bar{k}}\left[1+(N-2)N^{K-k-i}\right]+ \left(1+\frac{N-2}{N}\right)+ \left[H(A_{1:N-1}^{[\bar{k}+1]}S_N|\bF W_{1:\bar{k}+1}) +\frac{(N-2)}{N}H(A_{1:N}^{[\bar{k}+1]}|\bF W_{1:\bar{k}+1})\right]   \label{flat-bound-case2-2} \\
&\geq \sum_{i=1}^{\bar{k}}\left[1+(N-2)N^{K-k-i}\right]+ \left(1+\frac{N-2}{N}\right)\left(1+\frac{1}{N}+\cdots+\frac{1}{N^{k-\bar{k}-1}}\right)  \nonumber \\
&\qquad +(N-1)\sum_{i=1}^{k-\bar{k}-1}\frac{1}{N^i}H(A_{1:N-1}^{[\bar{k}+i]}S_N|\bF W_{1:\bar{k}+i}) \nonumber \\
&\qquad +\frac{1}{N^{k-\bar{k}-1}}\left[H(A_{1:N-1}^{[k]}S_N|\bF W_{1:k}) +\frac{(N-2)}{N}H(A_{1:N}^{[k]}|\bF W_{1:k})\right]     \label{flat-bound-case2-5} \\
&= \left(\bar{k}+\frac{(N-2)(N^{K-k}-1)}{N-1}\right)+ \frac{2(N^{k-\bar{k}}-1)}{N^{k-\bar{k}}} +(N-1)\sum_{i=1}^{k-\bar{k}-1}\frac{1}{N^i}H(A_{1:N-1}^{[\bar{k}+i]}S_N|\bF W_{1:\bar{k}+i}) \nonumber \\
&\qquad +\frac{1}{N^{K-\bar{k}-1}}\left[H(A_{1:N-1}^{[k]}S_N|\bF W_{1:k}) +N^{K-k-1}(N-2)H(A_{1:N}^{[k]}|\bF W_{1:k})\right]    \nonumber \\
&\qquad +\frac{N^{K-k}-1}{N^{K-\bar{k}-1}}H(A_{1:N-1}^{[k]}S_N|\bF W_{1:k})   \label{flat-bound-case2-6} \\
&\geq \left(\bar{k}+\frac{(N-2)(N^{K-k}-1)}{N-1}\right)+ \frac{2(N^{k-\bar{k}}-1)}{N^{k-\bar{k}}} +(N-1)\sum_{i=1}^{k-\bar{k}-1}\frac{1}{N^i}\Big(B_N(K-(\bar{k}+i-1),N-1)-1\Big) \nonumber \\
&\qquad +\frac{1}{N^{k-1}}\left[\frac{(N^{K-k}-1)(N-2)}{N(N-1)}+(K-k)\right] +\frac{N^{K-k}-1}{N^{k-1}}\Big(B_N(K-k+1,N-1)-1\Big),    \label{flat-bound-case2}
\end{align}
where \eqref{flat-bound-case2-1} follows by applying \Cref{lemma-iteration-2term} with $j=N-1$ for $1,2,\cdots,\bar{k}-1$ successively, 
and \eqref{flat-bound-case2-2}-\eqref{flat-bound-case2-5} are obtained by applying \Cref{lemma-iteration-2term} with $j=N-1$ and $q=0$ for $\bar{k},\bar{k}+1,\cdots,k-1$ successively. 
This proves the theorem. 

\section{Proof of \Cref{lemma-iteration-2term}}\label{proof-lemma-iteration-2term}
For any $k\in[1:K]$, $j\in[1:N-1]$, and non-negative integer $q$, we prove the lemma as follows,
\begin{align}
&H(A_{1:j}^{[k]}S_{j+1:N}|\bF W_{1:k})+N^q(j-1)H(A_{1:N}^{[k]}|\bF W_{1:k})  \nonumber \\
&=H(A_{1:j}^{[k]}S_{j+1:N}|\bF W_{1:k})+N^q(j-1)H(A_{1:N}^{[k]}|\bF W_{1:k})+N^q(j-1)H(S_{j+1:N}|\bF W_{1:k})  \nonumber \\
&\qquad -N^q(j-1)H(S_{j+1:N}|\bF W_{1:k})    \label{pf-flat-iteration-1}\\
&\geq \big(1+N^q(j-1)\big)H(A_{1:j}^{[k]}S_{j+1:N}|\bF W_{1:k})+N^q(j-1)H(A_{j+1:N}^{[k]}|\bF W_{1:k}) -N^q(j-1)H(S_{j+1:N}|\bF W_{1:k})  \label{pf-flat-iteration-2} \\
&\geq \sum_{i=1}^jH(A_i^{[k]}S_{j+1:N}|\bF W_{1:k})+N^{q-1}(j-1)\sum_{i=1}^NH(A_i^{[k]}|\bF W_{1:k})  \nonumber \\
&\qquad +\big(1+N^q(j-1)-j\big)H(A_{1:j}^{[k]}S_{j+1:N}|\bF W_{1:k}) -N^q(j-1)H(S_{j+1:N}|\bF W_{1:k})    \label{pf-flat-iteration-3} \\
&\geq \sum_{i=1}^jH(A_i^{[k+1]}S_{j+1:N}|\bF W_{1:k})-(j-1)H(S_{j+1:N}|\bF W_{1:k})+N^{q-1}(j-1)\sum_{i=1}^NH(A_i^{[k+1]}|\bF W_{1:k})  \label{pf-flat-iteration-4} \\
&\geq \left[H(A_{1:j}^{[k+1]}S_{j+1:N}|\bF W_{1:k})+(j-1)H(S_{j+1:N}|\bF W_{1:k})\right]-(j-1)H(S_{j+1:N}|\bF W_{1:k})  \nonumber \\
&\qquad +N^{q-1}(j-1)H(A_{1:N}^{[k+1]}|\bF W_{1:k})   \label{pf-flat-iteration-5} \\
&\geq H(A_{1:j}^{[k+1]}S_{j+1:N}|\bF W_{1:k})+N^{q-1}(j-1)H(A_{1:N}^{[k+1]}|\bF W_{1:k})    \\
&= \left(1+N^{q-1}(j-1)\right)+\left[H(A_{1:j}^{[k+1]}S_{j+1:N}|W_{1:k+1})+N^{q-1}(j-1)H(A_{1:N}^{[k+1]}|\bF W_{1:k+1})\right], \label{pf-flat-iteration-6}
\end{align}
where \eqref{pf-flat-iteration-2} and \eqref{pf-flat-iteration-5} follow from submodularity, 
and \eqref{pf-flat-iteration-4} follows from the privacy, the nonnegativity of $q$ and the inequality 
\begin{equation}
H(A_{1:j}^{[k]}S_{j+1:N}|\bF W_{1:k})\geq H(S_{j+1:N}|\bF W_{1:k}).
\end{equation}

\end{appendices}


\bibliographystyle{IEEEtran}
\bibliography{S-D-tradeoff_ref}

\end{document}